\RequirePackage{amsmath}

\documentclass[a4paper]{article}   

\usepackage{hyperref}
\usepackage{algorithmicx}
\usepackage[utf8]{inputenc}
\usepackage{algorithm}
\usepackage{algpseudocode}
\usepackage[margin=1in]{geometry}
\usepackage{graphicx}
\usepackage{float}
\usepackage{subfig}
\usepackage{amsthm,color}
\usepackage{amsmath,amssymb}
\usepackage[dvipsnames]{xcolor}
\usepackage{paralist}
\usepackage[title]{appendix}
\usepackage{systeme}

\newcommand{\rmv}[1]{}

\newcommand{\MP}[1]{}

\newcommand{\RMP}[1]{}

\newcommand{\VMP}[1]{}

\newcommand{\REG}{\textbf{R}}

\newcommand{\wu}{\textsc{Write}}
\newcommand{\ru}{\textsc{Read}}

\newcommand{\W}{\textsc{w}}
\newcommand{\R}{\textsc{r}}

\newcommand{\AW}{\mathcal{A}}

\newcommand{\pv}{\textrm{communication}}

\newcommand{\MM}{\mathcal{M}}

\newcommand{\UM}{\mathcal{U}}

\newcommand{\wrt}[1]{\mathit{write}(#1)}
\newcommand{\rd}[1]{\mathit{read}(#1)}

\newcommand{\rinit}{u_0}

\newcounter{NewCounter}
\newcounter{ReallyNewCounter}

\newtheorem{property}[ReallyNewCounter]{Property}
\newtheorem{theorem}[NewCounter]{Theorem}
\newtheorem{claim}{Claim}[NewCounter]
\newtheorem{corollary}[NewCounter]{Corollary}
\newtheorem{lemma}[NewCounter]{Lemma}
 \newtheorem{observation}[NewCounter]{Observation}
 \newtheorem{assumption}[NewCounter]{Assumption}
\newtheorem{definition}[NewCounter]{Definition}
\algdef{SE}[SUBALG]{Indent}{EndIndent}{}{\algorithmicend\ }
\algtext*{Indent}
\algtext*{EndIndent}

\title{On Atomic Registers and Randomized Consensus\\ in M\&M Systems\footnote{This work is an extension of results presented in preliminary form in~\cite{vassos2019}.}}

\author{Vassos Hadzilacos \qquad Xing Hu \qquad Sam Toueg\\~\\Department of Computer Science\\University of Toronto\\Canada}

\date{\today}
\begin{document}
\maketitle 
\begin{abstract}
Motivated by recent distributed systems technology,
	Aguilera \emph{et al.} introduced a hybrid model of distributed computing,
	called \emph{message-and-memory model} or \emph{m\&m model} for short~\cite{aguilera-etal18}.
In this model, processes can communicate by message passing
	and also by accessing some shared memory (e.g., through some RDMA connections).
We first consider the basic problem of implementing an atomic single-writer multi-reader (SWMR)
	register shared by \emph{all} the processes in m\&m systems.
Specifically, we give an algorithm that implements such a register in m\&m systems
	and show that it is optimal in the number of process crashes that it can tolerate.
This generalizes the well-known implementation of an atomic SWMR register
	in a pure message-passing system~\cite{attiya1995sharing}.
We then combine our register implementation for m\&m systems
	with the well-known randomized consensus algorithm of Aspnes and Herlihy~\cite{AH1990consensus},
	and obtain a randomized consensus algorithm for m\&m systems that is
	also optimal in the number of process crashes that it can tolerate.
%
%
Finally, we determine the minimum number of RDMA connections that is sufficient
	to implement a SWMR register, or solve randomized consensus, in an m\&m system with $t$ process crashes, for any given $t$.
\end{abstract}

\section{Introduction}
Motivated by recent distributed systems technology~\cite{genz,genzDRAM,rdma,iwarp,disaggmem,roce},
	Aguilera \emph{et al.} introduced a hybrid model of distributed computing,
	called \emph{message-and-memory model} or \emph{m\&m model} for short~\cite{aguilera-etal18}.
In this model processes can communicate by message passing
	and also by accessing some shared memory.
Since it is impractical to share memory among \emph{all} processes in large distributed systems \cite{farm1,herd,fasst,lite2017},
	the m\&m model allows us to specify 
	which subsets of processes share which sets of registers.
Among other results,
	Aguilera \emph{et al.} show that it is possible to leverage the advantages
	of the two communication mechanisms (message-passing and share-memory)
	to improve the fault-tolerance of randomized consensus algorithms
	compared to a pure message-passing system.

In this paper, we first consider the basic problem of implementing an atomic single-writer multi-reader (SWMR)
	register shared by \emph{all} the processes in m\&m systems, and 
	we give an algorithm that is optimal in the number of process crashes that it can tolerate.
This generalizes the fundamental ABD algorithm of Attiya, Bar-Noy, and Dolev that implements an atomic SWMR register
	in a pure message-passing system~\cite{attiya1995sharing}.
We then combine our register implementation for m\&m systems
	with the randomized consensus algorithm of Aspnes and Herlihy~\cite{AH1990consensus},
	and obtain a randomized consensus algorithm for m\&m systems that is
	also optimal in the number of process crashes that it can tolerate.
We now describe our results in more detail.

A \emph{general} m\&m system $\mathcal{S}_L$
	is specified by a set of $n$ asynchronous processes 
	that can send messages to each other over asynchronous reliable links,
	and by a collection $L=\{S_1, S_2,\dots, S_m\}$ where each $S_i$ is a subset of processes:
	for each $S_i$, there is a set of atomic registers that can be shared by processes in $S_i$ and only by them.
Even though the m\&m model allows the collection $L$ to be arbitrary,
	in practice hardware technology 
 	imposes a structure on $L$ \cite{farm1,herd}:
	for processes to share memory,
	they must establish a connection between them (e.g., an RDMA connection).
These connections are naturally modeled by an undirected \emph{shared-memory graph} $G$
	whose nodes are the processes and whose edges are shared-memory connections~\cite{aguilera-etal18}.
Such a graph $G$ defines what Aguilera \emph{et al.} call a \emph{uniform} m\&m system $\mathcal{S}_G$,
	where each process has atomic registers that it can share with its neighbours in $G$ (and only with them).
Note that $\mathcal{S}_G$ is the instance of the general m\&m system
	$\mathcal{S}_L$ with
	$L=\{S_1, S_2,\dots, S_n\}$ where
	each $S_i$ consists of a process and its neighbours in $G$.
Furthermore, if $G$ is the trivial graph with $n$ nodes but no edges,
	the m\&m system $\mathcal{S}_G$ that $G$ induces
	is just a pure message-passing asynchronous system with $n$ processes.

We consider the implementation of an atomic SWMR
	register $\REG$, shared by all the processes, in both general and uniform m\&m systems.
For each general m\&m system $\mathcal{S}_L$,
	we determine the maximum number of crashes $t_L$
	for which it is possible to implement $\REG$~in~$\mathcal{S}_L$:
	we give an algorithm that tolerates $t_L$ crashes
	and prove that no algorithm can tolerate more than $t_L$ crashes.
Similarly, for each shared-memory graph $G$ and its corresponding uniform m\&m system $\mathcal{S}_G$,
	we use the topology of $G$
	to determine the maximum number of crashes~$t_G$
	for which it is possible to implement $\REG$ in $\mathcal{S}_G$.
By specifying $t_G$ in terms of the topology of $G$,
	one can leverage results from graph theory to design m\&m systems
	that can implement $\REG$ with high fault tolerance
	and relatively few RDMA connections per process.
For example, it allows us to design an m\&m system with 50 processes
	that can implement a \emph{wait-free} $\REG$ (i.e., this implementation can tolerate \emph{any} number of process crashes)
	with only 7 RDMA connections per process;
	as explained in Section~\ref{noname}, this is optimal in some precise sense.

We then show how to solve randomized consensus in m\&m systems with optimal fault-tolerance.\footnote{This algorithm tolerates more failures than the one given for m\&m systems in~\cite{aguilera-etal18}.}
	This algorithm is derived in a simple way:
	we just substitute the atomic SWMR registers used by
	a known randomized consensus algorithm
	with our implementation of such registers for m\&m systems.
It is \emph{not} obvious, however, that one can always obtain a correct randomized consensus
	this way: it was recently shown that replacing the atomic registers of a randomized algorithm with register implementations
	that are \emph{only} linearizable
	may kill the termination property of that algorithm against a strong adversary~\cite{sltermination}.
So here we use a \emph{specific}
	randomized consensus algorithm, namely the one by Aspnes and Herlihy in~\cite{AH1990consensus},
	because it was shown that this algorithm works against a strong adversary even with \emph{regular} registers~\cite{hadzilacos2020randomized}.

We conclude the paper by determining the minimum number of RDMA connections required to
	achieve any desired degree of fault-tolerance
	when implementing SWMR registers or solving consensus in uniform m\&m systems.
Note that \emph{without any} RDMA connections, i.e., in a pure message-passing system,
	one can implement a SWMR atomic register,
	and solve randomized consensus, for up to $\lceil \frac{n}{2}\rceil -1$ process crashes (where $n$ is the number of processes).
We show here that the minimum number of RDMA connections required
	to tolerate $t > \lceil \frac{n}{2}\rceil -1$ crashes in a uniform m\&m system is simply $t$.

An important remark is now in order.
In this paper we consider RDMA systems where process crashes
	do not affect the accessibility of the shared registers of that system.
This is the case in systems where the CPU, the DRAM (main memory), and the NIC
	(Network Interface Controller) are separate entities:
	for example, in the InfiniBand cluster evaluated in~\cite{DARE}, the crash of a CPU, and of the processes that it hosts,
	does not prevent other processes from accessing its DRAM
	because it can use the NIC without involving the CPU; see also~\cite{farm1,genzDRAM,Orion}.

\section{Model outline}\label{model}
We consider m\&m systems with a set of $n$ asynchronous processes $\Pi = \{ p_1 , p_2 , \ldots,  p_n \}$ that may \emph{crash}.
To define these systems,
	we first recall the definition of atomic SWMR registers,
	and what it means to implement such registers.

\subsection{Atomic SWMR registers}

A register $R$ is \emph{atomic} if its read and write operations are \emph{instantaneous} (i.e., indivisible);
	each read 
	must return the value of the last write 
	that precedes it, or the initial value of $R$ if no such write exists.
A SWMR register $R$ is \emph{shared by a set $S \subseteq \Pi$ of processes} if it can be written
	(sequentially) by exactly one process $w \in S$
	and can be read by all processes~in~$S$; we say that $w$ is the \emph{writer} of $R$~\cite{Lam86}.	

\subsection{Implementation of atomic SWMR registers}\label{SWMR-Properties}

As in~\cite{attiya1995sharing}, we are interested in implementing atomic SWMR registers.
By implementation, we mean a \emph{linearizable} implementation of such registers,
	as we now explain.
In an object implementation,
	each operation spans an interval that starts
	with an \emph{invocation} and terminates with a \emph{response}.

\begin{definition}
Let $o$ and $o'$ be any two operations.
\begin{itemize}

\item $o$ \emph{precedes} $o'$ if the response of $o$ occurs 
	before \mbox{the invocation of $o'$.}

\item $o$ \emph{is concurrent with} $o'$ if neither precedes the other.

\end{itemize}
\end{definition}

Roughly speaking,
	an object implementation is \emph{linearizable}~\cite{HerlihyWing1990}
	if, although operations can be concurrent,
	operations behave as if they occur in a \emph{sequential} order (called ``linearization order'')
	that is consistent with the order in which operations actually occur:
	if an operation $o$ precedes an operation $o'$, then $o$ is before $o'$ in the linearization order
	(the precise definition is given in~\cite{HerlihyWing1990}).
	
It is well-known that linearizable implementations of atomic SWMR registers
	are characterized by two simple properties.
To define these properties, assume, without loss of generality, that the values successively
	written by the single writer $w$ of a SWMR register $R$ are distinct,
	and different from
	the initial value of $R$.\footnote{This can be ensured by the writer $w$
	writing values of the form $\langle sn , v \rangle$ where $sn$
	is the value of a counter that $w$ increments before each write.}
Let $v_0$ denote the \emph{initial value} of $R$,
	and $v_k$ denote the value written by the $k$-th write operation of $w$.
We say that a write operation \emph{$w$ immediately precedes a read operation $r$}
	if $w$ precedes $r$, and there is no write operation $w'$
	such that $w$ precedes $w'$ and $w'$ precedes $r$.
An atomic SWMR register implementation
	is linearizable if and only if it satisfies the following two properties.

\begin{property}
\label{p1}
If a read operation $r$ returns the value $v$
	then:
	
\begin{itemize}

\item there is a write $v$ operation that immediately precedes $r$ or is concurrent with $r$, or	
\item no write operation precedes $r$ and $v = v_0$.
\end{itemize}
\end{property}

\begin{property}
\label{p2}
If two read operations $r$ and $r'$ return values $v_k$ and $v_{k'}$, respectively,
	and  $r$ precedes $r'$, then $k \le k'$.
\end{property}

Henceforth, by ``implementation of an atomic SWMR register'',
	we mean a \emph{linearizable} implementation of such a register,
	i.e., one that satisfies the above \mbox{two~properties.}

\subsection{m\&m systems}
In this section, we define three types of m\&m systems.
In all models processes can communicate by sending messages and also by sharing some objects.
In the first model, processes can share objects of arbitrary types;
	in the second one, they can share any type of atomic registers (e.g., SWMR or MWMR registers);
	and in the third model, they can only share SWMR atomic registers.
More precisely, let $L=\{S_1, S_2,\dots, S_m\}$ be any bag
	of non-empty subsets of $\Pi = \{ p_1 , p_2 , \ldots, p_n \}$.

\begin{definition}\label{UniversalModel}
$\UM_{L}$ is the class of m\&m systems (induced by $L$), each consisting of:

\begin{enumerate}
\item\label{UMuno} The processes in $\Pi$.

\item\label{UMdue} Reliable asynchronous communication links between every pair of processes in~$\Pi$.

\item\label{UMtre}
The following set of shared objects:
	For each subset of processes $S_i$ in $L$, a non-empty set of objects that are shared
	by the processes in~$S_i$ (and only by them).
\end{enumerate}

\end{definition}

\begin{definition}\label{GrandModel}
$\MM_{L}$ is the class of m\&m systems (induced by $L$), each consisting of:

\begin{enumerate}
\item\label{Muno} The processes in $\Pi$.

\item\label{Mdue} Reliable asynchronous communication links between every pair of processes in~$\Pi$.

\item\label{Mtre} The following set of registers:
For each subset of processes $S_i$ in $L$, a non-empty set of atomic registers that are shared
	by the processes in $S_i$ (and only by them).
	
\end{enumerate}

\end{definition}

Note that $\MM_{L}$ includes m\&m systems that differ by the type and number of registers shared by the processes in each $S_i$.
Since we are interested in implementing atomic SWMR registers (shared by \emph{all} processes in the system),
	here we focus on an m\&m system of $\MM_{L}$
	in which the set of registers shared by the processes in each set $S_i$ are atomic SWMR registers.
More precisely, we focus on the m\&m system~$\mathcal{S}_{L}$ defined below:

\begin{definition}\label{gmm}
The general m\&m system $\mathcal{S}_{L}$ (induced by $L$) consists of:

\begin{enumerate}
\item\label{uno} The processes in $\Pi$.

\item\label{due} Reliable asynchronous communication links between every pair of processes in~$\Pi$.

\item\label{tre} The following set of registers:
	For each subset of processes $S_i$ in $L$ and each process $p \in S_i$, an atomic SWMR register,
	denoted $R_i[p]$, that can be written by $p$ and read by all processes in $S_i$ (and only by them).

\end{enumerate}

\end{definition}

In this paper, for every $L$, we give an algorithm that
	implements atomic SWMR registers shared by all processes
	in the m\&m system $\mathcal{S}_{L}$,
	and we show that this algorithm is optimal in the number of process crashes that can
	be tolerated.
In fact we prove a stronger result, any algorithm that implements such registers in \emph{any}
	m\&m system in $\UM_{L}$
	(where processes can shared arbitrary object, not just registers)
	cannot tolerate more crashes.

Without loss of generality we assume the following:

\begin{assumption}\label{singleton}
The bag $L=\{S_1, S_2,\dots, S_m\}$ of subsets of
	$\Pi = \{ p_1 , p_2 , \ldots, p_n \}$ is such that every process in $\Pi$ is in at least one of the subsets $S_j$ of $L$.
\end{assumption}

This assumption can be made without loss of generality
	because it does not strengthen the system $\mathcal{S}_{L}$ induced by $L$.
In fact, given a bag $L$ that does \emph{not} satisfy the above assumption,
	we can construct a bag that satisfies the assumption as follows:
	for every process $p_i$~in~$\Pi$ that is not contained in any $S_j$ of $L$,
	we can add the singleton set $\{ p_i \}$  to $L$.
Let $L'$ be the resulting bag.
By Definition~\ref{gmm}(\ref{tre}) above, adding  $\{ p_i \}$ to $L$ results in adding only a \emph{local} register to the induced system $\mathcal{S}_{L}$, namely, an atomic register that $p_i$ (trivially) shares only with itself.
So $\mathcal{S}_{L'}$ is just $\mathcal{S}_{L}$ with some additional local registers.
Note that a pure message-passing system (with no shared memory) with $n$ processes $p_1 , p_2 , \ldots, p_n$
	is modeled by the system $\mathcal{S}_{L}$ where
	$L = \{ \{ p_1 \} , \{ p_2 \} , \ldots ,\{ p_n \} \}$.

\subsection{Uniform m\&m systems}

Let $G=(V,E)$ be an undirected graph such that $V = \Pi$, i.e., the nodes of $G$ are the $n$ processes $p_1 , p_2 , \ldots ,p_n$ of the system.
For each $p_i \in V$, let $N(p_i)=\{p_j~|~(p_i,p_j)\in E\}$ be the \emph{neighbours} of $p_i$ in $G$,
	and let $N^+(p_i) = N(p_i)\cup \{p_i\}$.

\begin{definition}\label{SG-def}
The uniform m\&m system $\mathcal{S}_{G}$ (induced by $G$) is the m\&m system $\mathcal{S}_{L}$
	where~$L =  \{  S_1, S_2,\dots, S_n \}$ with $S_i = N^+(p_i)$.\footnote{Note that $L$ satisfies Assumption~\ref{singleton} because each $S_i = N^+(p_i)$ contains $p_i$.} 

\end{definition}

The graph $G$ induces the uniform m\&m system $\mathcal{S}_{G}$ where processes can communicate by message passing (via reliable asynchronous communication links), and also by shared memory as follows:
	for each process $p_i$, and every neighbour $p$ of $p_i$ in $G$ (including $p_i$)
	there is an atomic SWMR register $R_i[p]$
	that can be written by $p$ and read by every neighbour of $p_i$ in~$G$ (including~$p_i$).
We can think of the registers $R_i[*]$ as being physically located in the DRAM of the host of $p_i$,
	and every neighbour of $p_i$ accessing these registers over its
	RDMA connection to $p_i$ (which is modeled by an edge of $G$).\footnote{As we mentioned in the introduction,
	we assume that the crash of $p_i$
	does not prevent the neighbours~of~$p_i$
	from accessing the shared registers $R_i[*]$.}

\begin{figure}
\centering
\begin{minipage}[b]{.35\textwidth}
  \centering
 \hbox{\vspace{1.1cm} \hspace{-.5cm} \includegraphics[width=.70\linewidth]{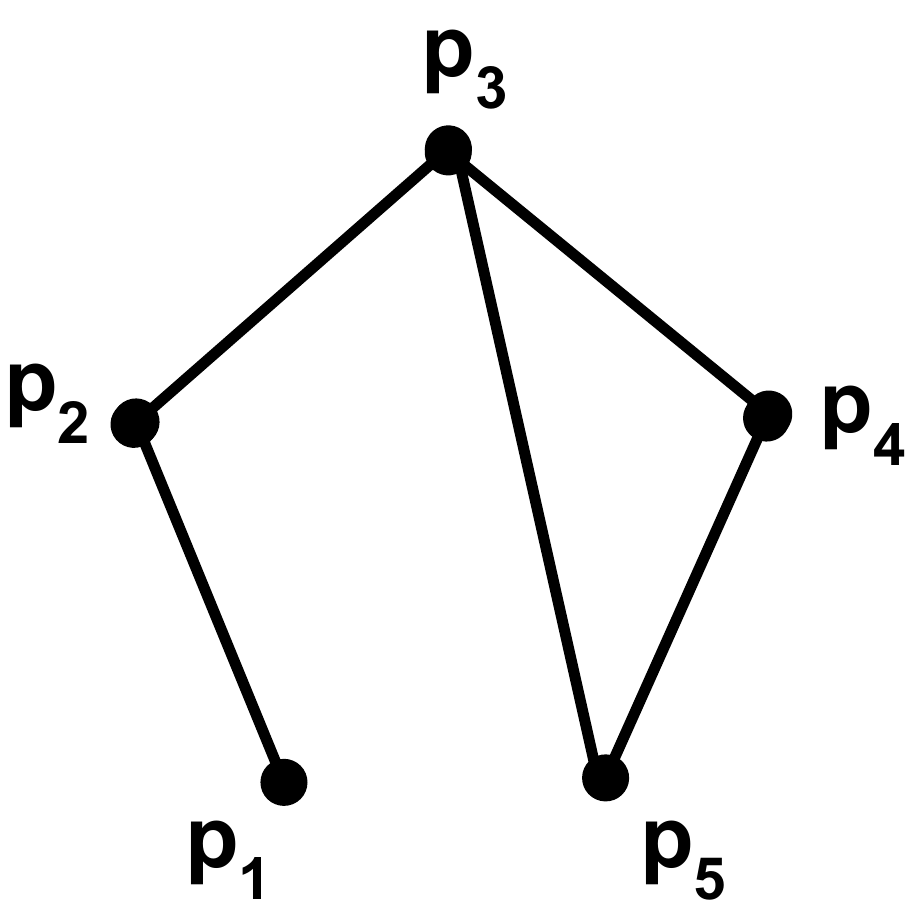}}
  \caption{A graph $G$}
  \label{G-fig}
\end{minipage}%
\begin{minipage}[b]{.6\textwidth}
  \centering
   \hbox{\hspace{-.75cm} \includegraphics[width=1.1\linewidth]{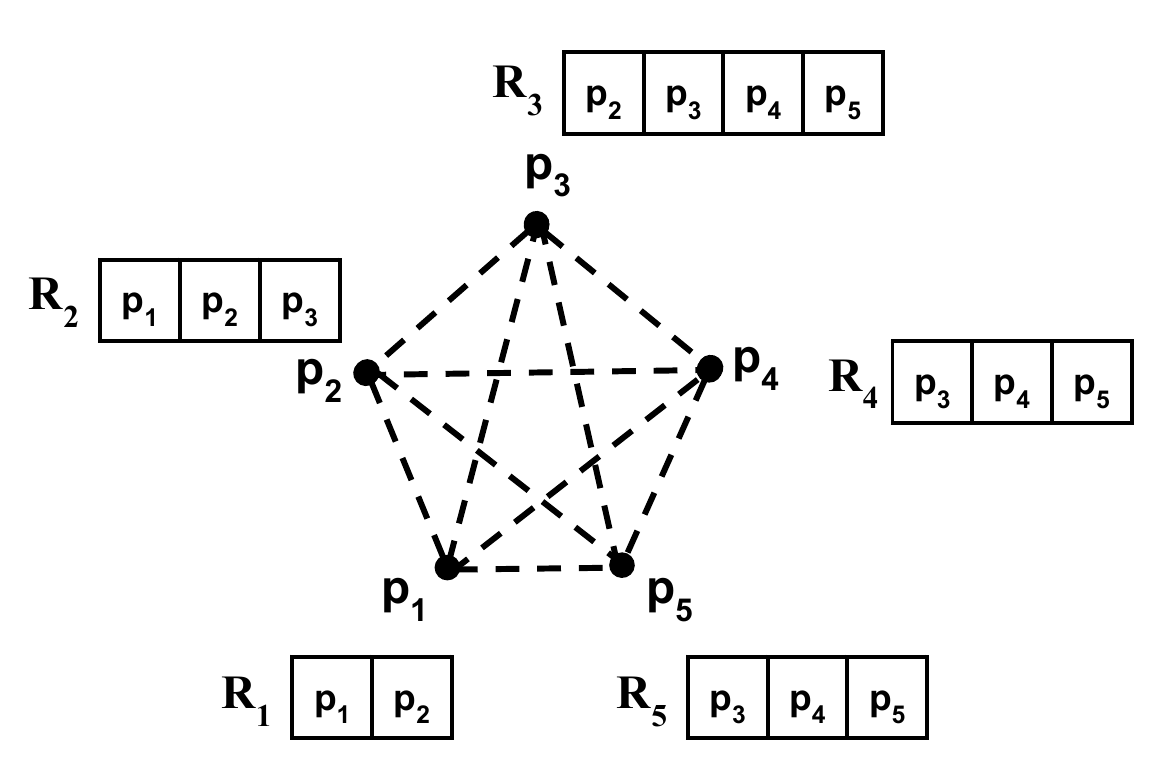}}
  \caption{The uniform m\&m system $\mathcal{S}_G$ induced by $G$}
  \label{SG-fig}
\end{minipage}
\end{figure}

For example, in Figures~\ref{G-fig} and~\ref{SG-fig}
	we show a graph $G$ and the uniform m\&m system $\mathcal{S}_G$ induced by $G$.
Here $G$ has five nodes representing processes $p_1, p_2, p_3, p_4, p_5$;
	the edges of $G$ represent the RDMA connections that allow these processes to share registers.
The uniform m\&m system $\mathcal{S}_G$ induced by $G$
	is the system $\mathcal{S}_L $ for $L= \{S_1, S_2, S_3, S_4, S_5\}$
	where each~$S_i$ consists of $p_i$ and its neighbours in $G$:
	specifically, $S_1 = \{p_1,p_2\}, S_2 = \{p_1,p_2,p_3\}$, $S_3 = \{p_2,p_3,p_4,p_5\}$,
	and $S_4 = S_5 = \{p_3,p_4,p_5\}$.
The box adjacent to each process $p_i$ in $\mathcal{S}_G$
	represents the atomic SWMR registers that are shared among $p_i$ and its neighbours in $G$
	(intuitively these registers are located at $p_i$'s host).
 For example, in the box adjacent to process $p_2$, the component labeled $p_1$ represents the register $R_2[p_1]$ that can be written by $p_1$ and read by all the neighbours of $p_2$ in $G$, namely $p_1,p_2$, and $p_3$.
 Similarly, registers $R_2[p_2]$ and $R_2[p_3]$ can be written by $p_2$ and $p_3$, respectively,
 	and read by $p_1,p_2$,~and~$p_3$.
The dashed lines in Figure~\ref{SG-fig} represent the asynchronous message-passing links between the~processes~of~$\mathcal{S}_G$.

\section{Atomic SWMR register implementation in general m\&m systems}\label{regsforgensys}

Let $\mathcal{S}_{L}$ be the general m\&m system induced by a bag
	$L = \{S_1,\dots ,S_m\}$ of subsets of $\Pi = \{ p_1 , p_2 , \ldots, p_n \}$.
Recall that in system $\mathcal{S}_{L}$, for every $S_i$ in $L$, the processes in $S_i$
	share some atomic SWMR registers that can be read \emph{only} by the processes in $S_i$.
In this section, we determine the maximum number of process crashes $t_L$
	that may occur in $\mathcal{S}_{L}$ such that
	it is possible to \emph{implement} in $\mathcal{S}_{L}$
	a shared atomic SWMR register readable by \emph{all} processes in $\mathcal{S}_L$.
Intuitively,
	from the definition of $t_L$:
	(a) if $t \le t_L$ processes may crash, then any two subsets of processes of size $n-t$ either intersect,
	or they each contain a process that can communicate with the other via a shared
	SWMR register that it can write and the other can read;
	and
	(b) if  $t > t_L$ processes may crash, then there are two subsets of processes of size $n-t$ that are disjoint
	and cannot communicate via  shared
	SWMR register.

\begin{definition}\label{def-tl}
Given a bag $L = \{S_1,\dots ,S_m\}$ of subsets of $\Pi = \{ p_1 , p_2 , \ldots, p_n \}$,
	$t_L$ is the maximum integer $t$ such that the following condition holds:
For all disjoint subsets $P$ and $P'$ of $\Pi$ of size $n-t$ each, 
	some set $S_i$ in $L$ contains both a process in $P$ and a process in $P'$.
\end{definition}

\noindent
Note that $t_L \le n-1$ because
	the maximum $t$ such that sets $P$ and $P'$ of size $n-t$ contain at least one node each
	must be less than $n$.
Moreover, if $t \le \lceil n/2\rceil -1$ then there are \emph{no} disjoint subsets $P$ and $P'$ of $\Pi$ of size $n-t$ each,
	and so the above condition is vacuously true.
Therefore $t_L  \ge \lceil n/2\rceil -1$.
Recall that for a pure message-passing system,
	$L = \{ \{ p_1 \} , \{ p_2 \} , \ldots, \{ p_n \} \}$, so in this system $t_L  = \lceil n/2\rceil -1$.

To illustrate Definition~\ref{def-tl},
	suppose $\Pi = \{ p_1 , p_2 , p_3, p_4 , p_5 \}$
	and $L = \{ S_1, S_2 , S_3 \}$
	where $S_1 = \{ p_1, p_2 \}$, $S_2 = \{ p_4, p_5 \}$, and $S_3 = \{ p_2, p_3, p_4 \}$.
By the definition of $t_L$:
	(1) $t_L \ge  3$ because for any two disjoint subsets of $\Pi$ of size $5 - 3 = 2$ each,
	there exists a set $S_i$ in $L$ that intersects both subsets;
	e.g., for subsets $\{ p_1, p_5\}$ and $\{p_3, p_4\}$, the set $S_2 = \{p_4, p_5 \}$ intersects both of them.
	(2) $t_L < 4$ because there are two disjoint subsets $\{ p_1 \}, \{p_5\}$ of size $5-4 = 1$ each,
	 such that no set $S_i$ in $L$ contains both $p_1$ and $p_5$.
So in this example $n = 5$ and $t_L = 3 >  \lceil n/2\rceil -1 =2$.

\smallskip
\noindent
We now prove that in the general m\&m system $\mathcal{S}_L$,
	it is possible to implement an atomic SWMR register readable by all processes
	\emph{if and only if} at most $t_L$ processes may crash~in~$\mathcal{S}_L$.	
More precisely:

\begin{theorem}\label{thm0}
Let $\mathcal{S}_{L}$ be the general m\&m system induced by a bag $L = \{S_1,\dots ,S_m\}$ of subsets of $\Pi = \{ p_1 , p_2 , \ldots, p_n \}$.

\begin{itemize}

\item If at most $t_L$ processes may crash in $\mathcal{S}_L$,
	then
	for every process $w$ in~$\mathcal{S}_L$,
	it is possible to implement an atomic SWMR register writable by $w$ and
	readable by all processes in~$\mathcal{S}_L$.

\item If $t_L+1 <n$ processes may crash in $\mathcal{S}_L$,
	then
	for every process $w$ in $\mathcal{S}_L$,
	it is impossible to implement an atomic SWMR register writable by $w$ and
	readable by all processes in~$\mathcal{S}_L$.

\end{itemize}

\end{theorem}

\noindent
The above theorem is a direct corollary of Theorem~\ref{thm-algo} (Section~\ref{algo}) and Theorem~\ref{thm-lb} (Section~\ref{lb}).

\subsection{Algorithm}\label{algo}

We now show how to implement an atomic SWMR register $\REG$,
 	that can be written by an arbitrary fixed process $w$
	and read by all processes,
	in an m\&m system $\mathcal{S}_{L}$ where at most $t_L$ processes may crash.
This implementation is given in terms of the procedures $\wu()$ and $\ru()$ shown in Algorithm~1.

Without loss of generality
	we assume that for all $i \ge 1$, the $i$-th value that the writer writes is of the form $\langle i , val \rangle$,
	and the initial value of the register $\REG$ is $\langle 0 , u_0\rangle$.
To write $\langle i , val \rangle$ into $\REG$, the writer $w$ calls the procedure $\wu(\langle i , val \rangle)$.
To read $\REG$, a process $q$ calls the procedure $\ru()$ that returns a value of the form  $\langle i, val\rangle$.
The sequence number $i$ makes the values written to $\REG$ unique.

\begin{algorithm}[!ht]
\caption{ Implementation of an atomic SWMR register writable by process~$w$ and readable by all processes in $ \mathcal{S}_L$, provided that at most $t_L$ processes crash.}
\label{algo1}
~\\
\textsc{Shared variables}

\vspace{2mm}

 For all $S_i \in L$ and all $p \in S_i$:

\vspace{1mm}
$R_i[p]:$ atomic SWMR register writable by $p$ and readable by every process in $S_i \in L$;\\
\hspace*{1.1cm}initialized to $\langle 0 , \rinit \rangle$.

\vspace{5mm}

\begin{algorithmic}[1]

\Statex

\textsc{\wu($\langle sn_w,u \rangle$):} \Comment{executed by the writer $w$}
\Indent
\State $\textbf{send } \langle \textrm{W},\langle sn_w, u\rangle \rangle \textbf{ to every process } p\mbox{ in } \mathcal{S}_L$\label{communicate}
\State \textbf{wait for} $\langle \textrm{ACK-W},sn_w \rangle \textbf{ messages from } n-t_{L} \textbf{ distinct processes }$\label{wait1}
\State \Return
\EndIndent
\Statex
\Statex

\Comment{executed by every process $p \mbox{ in } \mathcal{S}_L$}\\
\textbf{upon receipt of a }$\langle \textrm{W},\langle sn_w, u\rangle \rangle$ \textbf{message from process $w$}:
\Indent
\State \textbf{for every }{$i$ in $\{1,\dots,m\}$ such that $p\in S_i$}\textbf{ do}\label{for}
\Indent
\State\label{wr1} $\langle sn,-\rangle \gets R_i[p]$
\State\label{wr2} \textbf{if} {$sn_w> sn$}\textbf{ then}
\Indent
\State\label{wr3} $R_i[p]\gets \langle sn_w,u\rangle$\label{forend}
\EndIndent
\EndIndent
\State \textbf{send} $\langle\textrm{ACK-W}, sn_w \rangle$ \textbf{to process} $w$
\EndIndent
\Statex
\Statex

\textsc{\ru():}\Comment{executed by any process $q$}
\Indent
\State $sn_r\gets sn_r+1$
\State $\textbf{send } \langle \textrm{R}, sn_r\rangle \textbf{ to every process } p \mbox{ in } \mathcal{S}_L$
\State \textbf{wait for} $\langle\textrm{ACK-R}, sn_r, \langle -, - \rangle \rangle$ \textbf{messages from} $n-t_{L}$ \textbf{distinct processes}\label{wait2}
\State $\langle seq, val \rangle \gets \max\{ \langle r\_sn,r\_u \rangle~|$ received a $\langle\textrm{ACK-R}, sn_r,\langle r\_sn,r\_u \rangle \rangle$ message\} \label{adopt}
\State\label{callWUbyReader} \textsc{\wu}{$(\langle seq, val \rangle)$}
\State \Return $\langle seq, val \rangle$
\EndIndent
\Statex
\Statex

\Comment{executed by every process $p \mbox{ in } \mathcal{S}_L$}\\
\textbf{upon receipt of a} $ \langle \textrm{R}, sn_r\rangle$ \textbf{message from a process $q$}:
\Indent
\State $\langle r\_sn,r\_u \rangle  \gets \max\{\langle sn,u\rangle ~|~ \exists i\in \{1,\dots, m\}: p\in S_i$ and
	$\exists p' \in S_i : R_i[p']=\langle sn,u\rangle\}$\label{readall}
\State \textbf{send} $\langle \textrm{ACK-R}, sn_r,\langle r\_sn,r\_u \rangle\rangle$ \textbf{to process} $q$\label{sendtoreader}
\EndIndent
\Statex
\end{algorithmic}
\end{algorithm} 
 
Algorithm~\ref{algo1} generalizes the well-known ABD implementation of an atomic SWMR
	register in pure message-passing systems by Attiya, Bar-Noy and Dolev~\cite{attiya1995sharing}.
	\footnote{The ABD algorithm is the special case of Algorithm~\ref{algo1} for $\mathcal{S}_{L}$ where
	$L = \{ \{ p_1 \} , \{ p_2 \} , \ldots, \{ p_n \} \}$.}
To write a new value into $\REG$, the writer $w$ sends messages to all processes asking them to
	write the new value into all the shared SWMR registers that they can write in $\mathcal{S}_{L}$.
The writer then waits for acknowledgments from  $n-t_{L}$ processes
	indicating that they have done so.
To read $\REG$, a process sends messages to all processes asking them
	for the most up-to-date value that they can find in all the shared SWMR registers
	that they can read.
The reader waits for $n-t_{L}$ responses,
	selects the most up-to-date value among them,
	writes back that value (using the same procedure that the writer uses),
	and returns it.
From the definition of $t_L$ it follows that every write of $\REG$ ``intersects''
	with every read of $\REG$ at some shared SWMR register of $\mathcal{S}_{L}$.
Note that since at most $t_L$ processes crash, the waiting mentioned above does not block any process. 

We now show that the procedure $\wu()$, called by the writer $w$,
	and the procedure $\ru()$, called by any process $q$ in $\mathcal{S}_{L}$,
	implement an atomic SWMR register $\REG$.
To do so, we show that the calls of $\wu()$\emph{ by $w$} and of $\ru()$ by any process satisfy Properties 1 and 2 of atomic SWMR registers given in Section~\ref{SWMR-Properties}.

\begin{definition}
The operation $\wrt{v}$ is the execution of $\wu(v)$ \textbf{by the writer~$w$}
	for some tuple $v = \langle sn_w , u \rangle$:
	this operation starts when $w$ calls $\wu(v)$
	and it completes if and when this call returns.
An operation $\rd{v}$ is an execution of $\ru()$ that returns $v$ to some process $q$:
	this operation starts when $q$ calls $\ru()$
	and it completes when this call returns $v$ to $q$.
\end{definition}

Let $v_0 = \langle 0 , u_0\rangle$ be the initial value of the implemented register $\REG$,
	and, for $k\ge1$, let $v_k = \langle k , - \rangle$ denote the $k$-th value written by the writer $w$ on $\REG$.
Note that all $v_k$'s are distinct: for all $i \neq j \ge 0, v_i \neq v_j$.

Let $\mathcal{S}_{L}$ be the general m\&m system induced by
	a bag $L = \{S_1,\dots ,S_m\}$ of subsets of $\Pi = \{ p_1 , p_2 , \ldots, p_n \}$.
To prove the correctness of the SWMR implementation shown in Algorithm 1,
	we now consider an arbitrary execution of this implementation in $\mathcal{S}_{L}$.

\begin{lemma}\label{l0}
 If at most $t_L$ processes crash,
	then any $\rd{-}$ or $\wrt{-}$ operation executed by a
	process that does not crash completes.
\end{lemma}

\begin{proof}
The only statements that could prevent the completion of a $\rd{-}$ or $\wrt{-}$ operation are
	the \textbf{wait} statements of line \ref{wait1} and line \ref{wait2}.
But since communication links are reliable, these wait statements are for $n - t_L$ acknowledgments,
	and at most $t_L$ processes out of the $n$ processes of $\mathcal{S}_{L}$ may crash,
	it is clear that these wait statements cannot block.
\qed
\end{proof}

The proofs of the next two lemmas are straightforward and therefore omitted.
The first one states that every read operation returns some $v_k$ for $k \ge 0$.

\begin{lemma}\label{l1}
If $r$ is a $\rd{v}$ operation in the execution,
	then $v=v_k$ for some $k \ge 0$.
\end{lemma}

The next lemma says that no read operation can read a ``future'' value, i.e., a  value that is written after the read completes.

\begin{lemma}\label{l2}
If $r$ is a $\rd{v}$ operation in the execution,
	then either $v = v_0$, or $v=v_k$ such that the operation $\wrt{v_k}$ precedes $r$ or is concurrent with $r$.
\end{lemma}

 Note that the guard in lines~\ref{wr2}-\ref{wr3} (which is the only place where the shared SWMR registers are updated)
 	ensures that
	the content of each shared SWMR register in $\mathcal{S}_{L}$
	is non-decreasing in the following sense:
	
\begin{observation}\label{un}
[Register monotonicity] For all $1 \le i \le m$ and all $p \in S_i$,
	if $R_i[p] = \langle k, - \rangle$ at some time $t$ and
	$R_i[p] = \langle k', - \rangle$ at some time $t' \ge t$
	then
	$k' \ge k$.
\end{observation}

\begin{lemma}\label{l3}
For all $k \ge 1$, if a call to the procedure $\wu(v_k)$ returns before a $\rd{v}$ operation starts,
	then $v = v_{\ell}$ for some $\ell \ge k$.
\end{lemma}

\begin{proof}
Suppose a call to $\wu(v_k)$ returns before a $\rd{v}$ operation starts;
	we must show that $v = v_{\ell}$ with $\ell \ge k$.
Note that before this call of $\wu(v_k)$ returns,
	$\langle \textrm{ACK-W}, k \rangle$ messages are received from a set $P$ of $n-t_L$ distinct processes
	(see line \ref{wait1} of the $\wu()$ procedure).
From lines~\ref{for}-\ref{wr3}, which are executed before
	 these $\langle \textrm{ACK-W}, k \rangle$ messages are sent,
	 and by Observation~\ref{un}, it is clear that the following holds:

 \begin{claim}\label{deux}
 By the time $\wu(v_k)$ returns,
	 every shared SWMR register in $\mathcal{S}_{L}$ that can be written by a process in $P$ contains
	 a tuple $\langle k', - \rangle$ with $k' \ge k$.
\end{claim}

Now consider the $\rd{v}$ operation, say it is by process $q$.
Recall that $\rd{v}$ is an execution of the $\ru()$ procedure that returns $v$ to $q$.
When $q$ calls $\ru()$, it increments a local counter $sn_r$
	and asks every process $p$ in $\mathcal{S}_{L}$ to do the following:
	(a) read every SWMR register that $p$ can read,
	and
	(b) reply to $q$ with a $\langle \textrm{ACK-R}, sn_r ,\langle r\_sn, r\_u \rangle\rangle$ message
	such that $\langle r\_sn, r\_u \rangle$ is the tuple with the maximum $r\_sn$ that $p$ read.
By line \ref{wait2} of the $\ru()$ procedure,
	$q$ waits to receive such $\langle \textrm{ACK-R}, sn_r, \langle  - , - \rangle\rangle$ messages
 	from a set $P'$ of $n-t_L$ distinct processes,
	and $q$ uses these messages to select the value $v$ as follows:
$$ v \gets \max\{ (r\_sn,r\_u)~|~ q \mbox{ received some } \langle\textrm{ACK-R}, sn_r,\langle r\_sn,r\_u \rangle \rangle \mbox{ from a process in } P' \}$$

\smallskip\noindent
Thus, by Lemma~\ref{l1}, it is clear that:

\begin{claim}\label{trois}
 $v= v_{\ell}$ where $\ell = \max \{ j  ~|$ $q$ received
	a $\langle\mbox{\emph{ACK-R}}, sn_r , \langle j ,- \rangle \rangle$ message from a process in $P' \}$.
\end{claim}

	\noindent
	
\rmv{
\hrule
Let $\mathcal{R}$ be the set of all the SWMR registers in $\mathcal{S}_{L}$ that can be read by processes in $P'$,
	i.e., $\mathcal{R} = \{ R_i[q] ~|~ 1 \le i \le m \mbox{ and } q \in S_i  \mbox{ and } p \in P' \cap S_i \}$.
Let $\mathcal{V}$ be the set of all the values of the registers in $\cal{R}$ that processes in $P'$ read
	during the execution of $\rd{v}$.
We claim that $v = v_{\ell}$ such that $\ell = \max\limits_{j} \{ v_j = \langle j , r\_u \rangle \in \cal{V} \}$.

\hrule
}

\begin{claim}\label{quatre}
 Some set $S_i$ in $L$ contains both a process in $P$ and a process in $P'$.
\end{claim}

\begin{proof}
If $P$ and $P'$ are disjoint, the claim follows directly from Definition~\ref{def-tl} of $t_L$.
If $P$ and $P'$ intersect,
	let $p$ be a process in both $P$ and $P'$.
By Assumption~\ref{singleton},
	$p$ is in some set $S_i$ in $L$, and the claim follows.
\qed
\end{proof}

 By the above claim, some set $S_i$ in $L$ contains a process $p$ in $P$ and a process $p'$ in $P'$.
Since $p\in S_i$ and $p' \in S_i$, $R_i[p]$ is one of the SWMR registers that can be written by $p$ and read by $p'$.
From Claim~\ref{deux},
	 by the time the call to $\wu(v_k)$ returns,
	$R_i[p]$ contains a tuple $\langle k', - \rangle$ such that $k' \ge k$ (*).
Since $p' \in P'$, during the execution of $\rd{v}$ by $q$, $p'$
	reads all the shared SWMR registers that it can read, including $R_i[p]$.
Since $\rd{v}$ starts after $\wu(v_k)$ returns, $p'$ reads $R_i[p]$ after $\wu(v_k)$ returns.
Thus, by (*) and the monotonicity of $R_i[p]$ (Observation~\ref{un}), $p'$ reads from $R_i[p]$ a tuple $\langle r\_sn, - \rangle$
	with $r\_sn \ge k' \ge k$.
Then $p'$ selects  the tuple $\langle j, - \rangle$ with the maximum $sn$ among all the $\langle sn, - \rangle$ tuples that it read
	(see line~\ref{readall});
	note that $j \ge k$.
So the $\langle\textrm{ACK-R}, sn_r , \langle j ,- \rangle \rangle$ message that $p'$ sends to $q$, and $q$ uses to select $v$,
	is such that $j \ge k$.
So, by Claim~\ref{trois}, $v = v_{\ell}$ such that $\ell \ge j \ge k$.
\end{proof}

Lemma~\ref{l3} immediately implies the following:

\begin{corollary}\label{c3}
For all $k \ge 1$, if a $\wrt{v_k}$ operation precedes a $\rd{v}$ operation then $v = v_{\ell}$ with $\ell \ge k$.
\end{corollary}

We now show that  Algorithm 1 satisfies Properties 1 and 2 of atomic SWMR registers.

\begin{lemma}\label{prop1}
The $\wrt{-}$ and $\rd{-}$ operations satisfy Property 1.
\end{lemma}

\begin{proof}
Suppose for contradiction that Property 1 does not hold.
Thus there is a read operation $r = \rd{v}$ such that:

\begin{enumerate}[(a)]
\item there is no $\wrt{v}$ operation that immediately precedes $r$ or is concurrent with $r$,
	and

\item  some $\wrt{-}$ operation precedes $r$, or $v \neq v_0$.
\end{enumerate}

There are two cases.

\begin{enumerate}
\item $v = v_0$. By (b) above, some $\wrt{-}$ operation, say  $\wrt{v_k}$, precedes $r$.
Thus $\wrt{v_k}$ precedes $\rd{v_0}$.
Since $k\ge 1$ this contradicts Corollary~\ref{c3}.

\item $v \neq v_0$.
By Lemma~\ref{l2}, $v= v_k$ such that the operation $\wrt{v_k}$ precedes $r$,
	or  $\wrt{v_k}$ is concurrent with $r$.
By (a) above, $\wrt{v_k}$ does not \emph{immediately} precede $r$,
	and $\wrt{v_k}$ is not concurrent with $r$.
 Thus, $\wrt{v_k}$ precedes, but not \emph{immediately}, $r$.
 Let $\wrt{v_{k'}}$ be the write operation that immediately precedes $r$.
 Note that $\wrt{v_k}$ precedes $\wrt{v_{k'}}$, so $k <k'$.
 Since $\wrt{v_{k'}}$  precedes $r= \rd{v}$,
 	by Corollary~\ref{c3}, $v=v_{\ell}$ with $\ell \ge k'$, so $\ell > k$.
This contradicts that $v = v_k$.
\end{enumerate}
Since both cases lead to a contradiction, Property 1 holds.
\end{proof}

\begin{lemma}\label{prop2}
The $\wrt{-}$ and $\rd{-}$ operations satisfy Property 2.
\end{lemma}

\begin{proof}
We have to show that if a $\rd{v_k}$ operation precedes a $\rd{v_{k'}}$ operation then $ k \le k'$.
Suppose $\rd{v_k}$ precedes $\rd{v_{k'}}$.
Note that during the $\rd{v_k}$ operation, namely in line~\ref{callWUbyReader}, there is a call to the procedure $\wu(v_k)$
	which returns before the $\rd{v_k}$ operation completes.
So this call to $\wu(v_k)$ returns before the $\rd{v_{k'}}$ operation starts.
By Lemma~\ref{l3}, $k \le k'$.
\end{proof}

Lemmas~\ref{l0}, \ref{prop1} and~\ref{prop2} immediately imply:

\begin{theorem}\label{thm-algo}
Let $\mathcal{S}_{L}$ be the general m\&m system induced by a bag $L = \{S_1,\dots ,S_m\}$ of subsets of $\Pi = \{ p_1 , p_2 , \ldots, p_n \}$.
If at most $t_L$ processes may crash in $\mathcal{S}_L$,
	for every process $w$ in~$\mathcal{S}_L$,
	Algorithm 1 implements an atomic SWMR register writable by $w$ and readable by all processes in~$\mathcal{S}_L$.
\end{theorem}

\subsection{Lower bound}\label{lb}

\begin{theorem}\label{thm-lb}
Let $\mathcal{S}_{L}$ be the general m\&m system induced by a bag $L =  \{S_1,\dots ,S_m\}$ of subsets of $\Pi = \{ p_1 , p_2 , \ldots, p_n \}$.
If $t_L +1 <n$ processes may crash in $\mathcal{S}_L$,
	then for every process $w$ in $\mathcal{S}_L$,
	there is no algorithm that implements an atomic SWMR register writable by~$w$ and readable by all processes in~$\mathcal{S}_L$.
\end{theorem}

\begin{proof}

Let $\mathcal{S}_{L}$ be the general m\&m system
	induced by a bag $L =  \{S_1,\dots ,S_m\}$ of subsets of $\Pi = \{ p_1 , p_2 , \ldots, p_n \}$.
Suppose for contradiction that
	$t  = t_L +1 <n$ processes may crash in $\mathcal{S}_L$,
	but
	for some process $w$ in $\mathcal{S}_L$,
	there is an algorithm $\AW$ that implements
	an atomic SWMR register writable by $w$
	and readable by all processes in~$\mathcal{S}_L$ (*).

\begin{figure}[!htb]
    \centering 
    \includegraphics[width=0.7\textwidth]{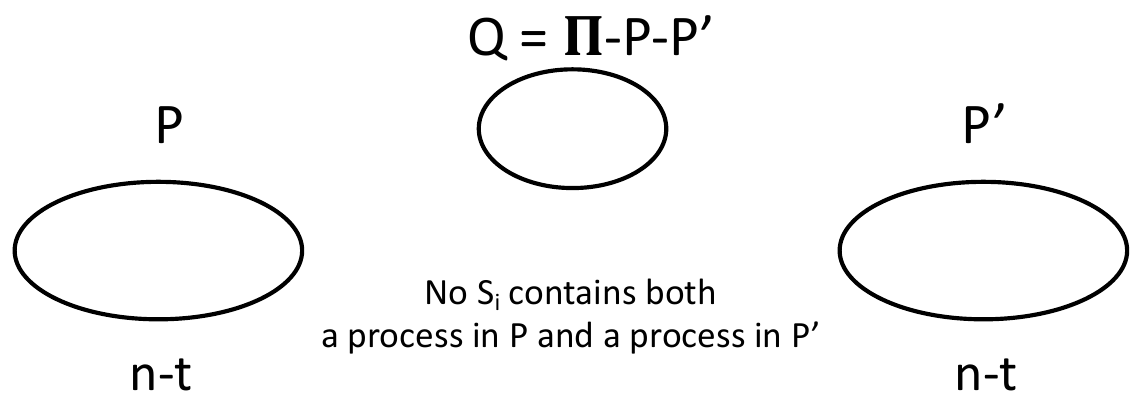}
    \caption{Partition of $\Pi$} 
    \label{gsx2}
\end{figure}

Since $t  > t_L$,
	by the Definition~\ref{def-tl} of $t_L$ there are two disjoint subsets $P$ and $P'$
	of $\Pi$,
	of size $n-t$ each, such that
	no set $S_i$ in $L$ contains both a process in $P$ and a process in $P'$~(**).
	
Since $P$ and $P'$ are disjoint, the sets $P$, $P'$, and $Q=\Pi-(P\cup P')$
	form a partition of $\Pi$ (see Figure~\ref{gsx2}).
Since $t <n$, each of $P$ and $P'$ contains at least one process, say $p \in P$ and $p' \in P'$.
Since $|P\cup Q \cup P'| = n$, clearly $|P\cup Q|= |P'\cup Q|= n - (n-t) =t$~($\dagger$).
Since algorithm $\AW$ tolerates $t$ crashes,
	it works correctly
	in every execution in which all the processes in $P\cup Q$ or in $P'\cup Q$ crash.
	
There are two cases.

\smallskip
\noindent
\textbf{Case 1:} $w\in P$ or $P'$.
Without loss of generality,
	assume $w\in P$.
	
 We now define three executions $E_1$, $E_2$, and $E_3$ of algorithm~$\mathcal{A}$.
These are illustrated in Figure \ref{g1}.

\begin{figure}[!htb]
    \centering 
    \includegraphics[width=0.5\textwidth]{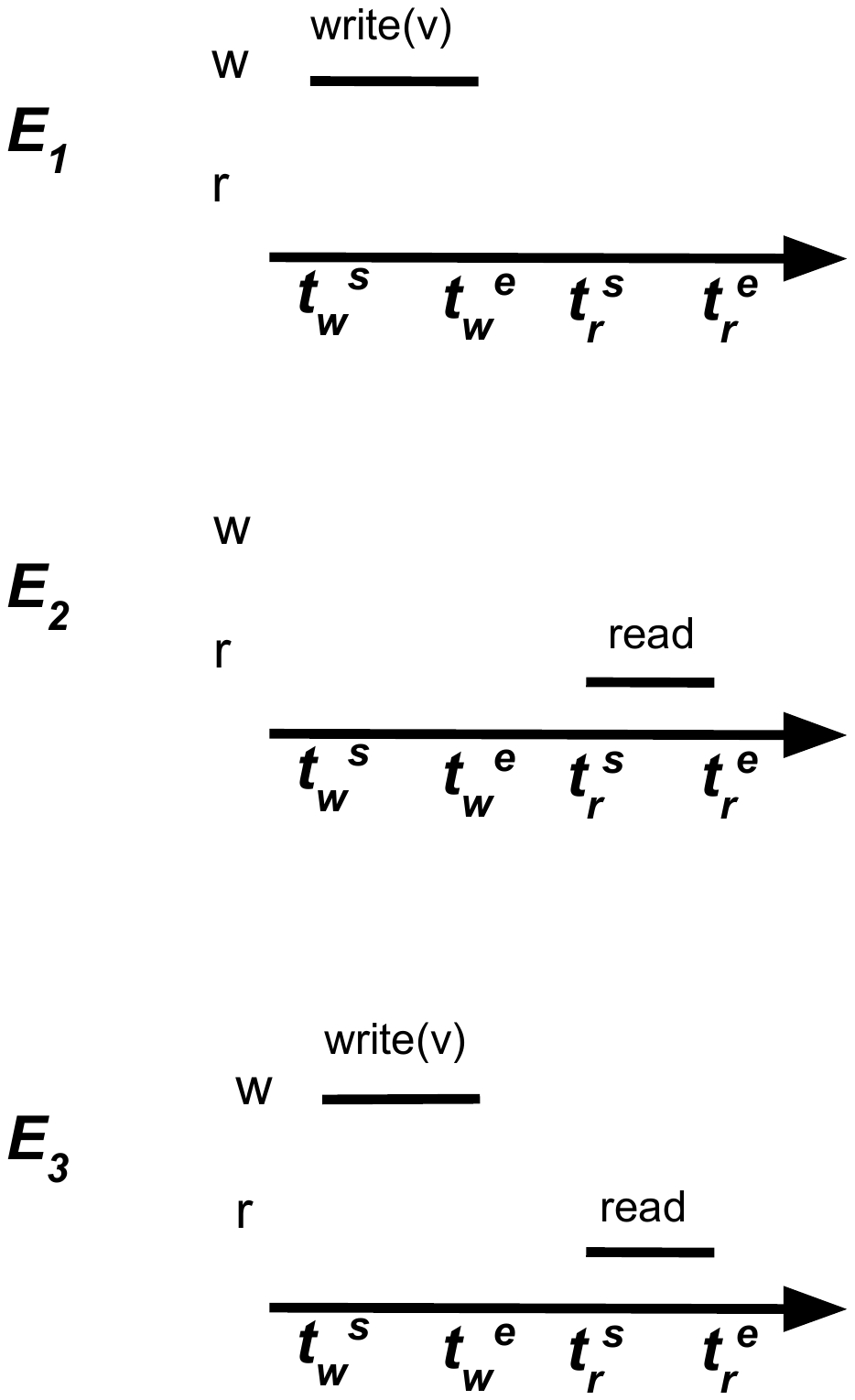}
    \caption{Scenarios for Theorem~\ref{thm-lb}} 
    \label{g1}
\end{figure}

Execution $E_1$ of algorithm $\mathcal{A}$ is defined as follows:
\begin{itemize}
\item The processes in $P' \cup Q$ crash from the beginning of the execution;
	they take no steps~in~$E_1$.
    
\item At some time $t_w^s$ the writer $w$ starts an operation to write the value $v$
	into the implemented register $\REG$, for some $v\ne v_0$,
	where $v_0$ is the initial value of $\REG$.
Since the number of processes that crash in $E_1$ is $|P'\cup Q|=t$,
	and the algorithm $\mathcal{A}$ tolerates $t$ crashes,
	this write operation eventually terminates,
	say at time~$t_w^e$.
	
\item After this write terminates, no process takes a step up to and including some time $t_r^s > t_w^e$.

\end{itemize}

\noindent
Note that in $E_1$, processes in $P$ are the only ones that take steps up to time $t_r^s$.

\smallskip
\noindent
Execution $E_2$ of algorithm $\mathcal{A}$ is defined as follows:
\begin{itemize}
\item The processes in $P \cup Q$ crash from the beginning of the execution;
	they take no steps~in~$E_2$.
	
\item Before time $t_r^s$, no process in $P'$ takes a step.

\item At time $t_r^s$,
	some process $r\in P'$
	starts a read operation on the implemented register~$\REG$.
Since the number of processes that crash in $E_2$ is $|P\cup Q|=t$,
	and the algorithm $\mathcal{A}$ tolerates $t$ crashes,
	this read operation terminates,
	say at time $t_r^e$.
\end{itemize}

Since no write operation precedes the read operation in $E_2$,
	Property~\ref{p1} of atomic SWMR registers implies:

\begin{claim}\label{E2}
At time $t_r^e$ in $E_2$
	the read operation returns the initial value $v_0$ of $\REG$.
\end{claim}

We now construct an execution $E_3$ of the algorithm $\mathcal{A}$ that merges $E_1$ and $E_2$,
	and contradicts the atomicity of the implemented $\REG$.
$E_3$ is identical to $E_1$ up to time $t_r^s$,
	and it is identical to $E_2$ from time $t_r^s$ to $t_r^e$ 
	(note that in $E_3$ processes in $Q$ can only take steps \emph{after} time $t_r^e$).
To obtain this merged run $E_3$, intuitively we delay
	the messages sent by processes in $P$ to processes in $P'$ until after time $t_r^e$,
	and we also use the fact that processes in $P'$ cannot read any of the shared registers
	in $\mathcal{S}_L$ that processes in $P$
	may have written by time $t_r^s$ (this is because of (**)).

\begin{claim}\label{lb-contra}
There is an execution $E_3$ of algorithm $\mathcal{A}$ such that
\begin{enumerate}[\noindent(a)]
\item up to and including time $t_w^e$, $E_3$ is indistinguishable from $E_1$
	to all processes.
\item up to and including time $t_r^e$, $E_3$  is indistinguishable from $E_2$
	to all processes in $P'$.
\item No process crashes in $E_3$.
\end{enumerate}
\end{claim}

\begin{proof}
Until time $t_r^s$, $E_3$ is identical to $E_1$.
We now show that it is possible to extend $E_3$ in the time interval $[t_r^s,t_r^e]$
	with the sequence of steps
	that the processes~in~$P'$ executed during the same time interval in $E_2$.\footnote{A \emph{step} of $\mathcal{A}$ executed by process $p$
		is one of the following:
		$p$ sending or receiving a message, or
		$p$ applying a write or a read operation to a shared register in $\mathcal{S}_L$.}
More precisely,
	let $s^1,s^2,\ldots,s^{\ell}$ be the sequence of steps
	executed during the time interval $[t_r^s,t_r^e]$ in $E_2$.
Since only processes in $P'$ take steps in $E_2$,
	$s^1,s^2,\ldots,s^{\ell}$ are all steps of processes in $P'$.
Let $C_2^0$ be the configuration of the system $\mathcal{S}_L$ at time $t_r^s$ in $E_2$,\footnote{The
	\emph{configuration} of $\mathcal{S}_L$ at time $t$ in execution $E$ consists of
	the state of each process,
	the set of messages sent but not yet received, and
	the value of each shared register in $\mathcal{S}_L$
	at time $t$ in $E$.}
	and let $C_2^i$ be the configuration that results
	from applying step $s^i$ to configuration $C_2^{i-1}$,
	for all $i$ such that $1\le i\le \ell$.
We will prove that
	there are configurations $C_3^0,C_3^1,\ldots,C_3^{\ell}$ of $\mathcal{S}_L$
	extending $E_3$ at time $t_r^s$ such that:
\begin{enumerate}[\noindent(i)]
\item every process in $P'$ has the same state in $C_3^i$ as in $C_2^i$;
\item the set of messages sent by processes in $P'$ to processes in $P'$,
	but not yet received, is the same in $C_3^i$ as in $C_2^i$;
\item every shared register readable by processes in $P'$
	has the same value in $C_3^i$ as in $C_2^i$; and
\item if $i\ne 0$, $C_3^i$ is the result of
	applying step $s^i$ to configuration $C_3^{i-1}$.
\end{enumerate}
This is shown by induction on $i$.

For the basis of the induction, $i=0$, we take $C_3^0$ to be the configuration of the system
	just before time $t_r^s$ in $E_3$.
Since no process in $P'$ takes a step before time $t_r^s$ in either $E_2$~or~$E_3$,
	$C_3^0$ satisfies properties~(i) and~(ii).

\begin{claim}\label{E3}
At time $t_r^s$ in $E_3$
	the shared registers that can be read by processes in $P'$
	have their initial values.
\end{claim}

\begin{proof}
Suppose, for contradiction, that at time $t_r^s$ in $E_3$,
	some shared register $R$ that can be read by a process $p'$ in $P'$
	does not have its initial value.
By construction, $E_3$ is identical to $E_1$ until time $t_r^s$, and so
	only processes in $P$ take steps before time $t_r^s$ in $E_3$.
Thus, register $R$ was written by some process $p$ in $P$ by time $t_r^s$ in $E_3$.
Since $R$ is readable by $p' \in P'$ and is written by $p \in P$, $R$ is shared by~both~$p$~and~$p'$.
Thus, there must be a set $S_i$ in $L$ that contains both $p$ and $p'$ --- a contradiction to (**).
\end{proof}

We now return to the proof of Claim \ref{lb-contra}. By Claim~\ref{E3}, the shared registers readable by processes in $P'$
	have the same value (namely, their initial value) in $C_3^0$ as in $C_2^0$.
So, $C_3^0$ also satisfies property~(iii).
Property~(iv) is vacuously true for $i=0$.

For the induction step, for each $i$ such that $1\le i\le\ell$,
	we consider separately the cases of $s^i$ being a step to
	send a message,
	receive a message,
	write a shared register, and
	read a shared register.
In each case, it is easy to verify that,
	assuming (inductively) that $C_3^{i-1}$ has properties (i)--(iv),
	step $s^i$ is applicable to $C_3^{i-1}$,
	and the resulting configuration $C_3^i$ has properties (i)--(iv).

To complete the definition of $E_3$,
	after time $t_r^e$
	we let processes take steps in round-robin fashion.
Whenever a process's step is to receive a message,
	it receives the oldest one sent to it; this ensures that all messages are eventually received.
Processes continue taking steps in this fashion
	according to algorithm~$\mathcal{A}$.

Since $E_3$ is identical to $E_1$ up to and including time $t_w^e$,
	$E_3$ is indistinguishable from $E_1$ up to and including time $t_w^e$
	to all processes in $P$.
This proves part~(a) of the claim.

Note that in $E_3$ and $E_2$, the processes in $P'$:
	(a)  take no steps before time $t_r^s$, and
	(b) during the time interval $[t_r^s,t_r^e]$,
	they execute exactly the same sequence of steps,
	and go through the same sequence of states.
Thus, up to and including time $t_r^e$, $E_3$ is indistinguishable from $E_2$
	to all processes in $P'$.
This proves part~(b) of the claim.

Finally, every process takes steps as required by the algorithm in $E_3$,
	so no process crashes.
This proves part~(c) of Claim \ref{lb-contra}.
\end{proof}

By Claim~\ref{lb-contra}(a),
	up to and including time $t_w^e$, $E_3$ is indistinguishable from $E_1$
	to the writer $w \in P$.
So $E_3$ contains the write operation that writes $v\ne v_0$ into $\REG$,
	which starts at time $t_w^s$ and ends at time $t_w^e$.
By Claim~\ref{lb-contra}(b),
	up to and including time $t_r^e$, $E_3$ is indistinguishable from $E_2$
	to $r \in P'$.
So $E_3$ contains the read operation that returns $v_0$,
	which starts at time $t_r^s$ and ends at time $t_r^e$.
Since $t_w^e<t_r^s$,
	this read operation violates Property~\ref{p1} of atomic SWMR registers.
As there are no process crashes in $E_3$ (by Claim~\ref{lb-contra}(c)),
	this contradicts the assumption
	that $\mathcal{A}$ is an implementation of an atomic SWMR register $\REG$
	that tolerates $t>t_L$ crashes.

\smallskip
\noindent
\textbf{Case 2:} $w \in Q$.

We now construct a sequence of executions of $\AW$ that leads to a contradiction.
In all these executions
	 all the processes in $Q$ except for $w$ are crashed from the start: they take no steps.

Let $E$ be the following execution of $\AW$ (Figure \ref{e}):

\begin{itemize}   

\item All the processes in $P'$ are crashed: they take no steps.

\item All the processes in $P$ are correct.

\item At some time $t_w^0$,
	$w$ starts an operation $write(v)$ to write $v\ne v_0$
	into the implemented register $\REG$,
	where $v_0$ is the initial value of $\REG$.
		
During this write operation, 
	$w$ executes the sequence of $\pv$ steps $s^1,...,s^k$,\MP{``$\pv$'' is a macro that can be changed} 
	say $s^i$ occurs at time $t_w^i$.
Recall that each step $s^i$ is one of the following:
	receiving messages,
	sending a message,
	reading a shared register, 
	or
	writing a shared register.
	
\item $w$ completes its $write(v)$ operation at time $t_w^{k+1}$. 

\item At some time $t_c> t_w^{k+1}$, process $w$ crashes.

Note that at this point all the processes in $P' \cup Q$ have crashed in $E$.
By ($\dagger$), this is a total of $t$ crashes.

\item At some time $r_s > t_c$, process $p$ starts reading $\REG$, and at time $r_e$ this operation
	completes and returns $v$.

\end{itemize}

\begin{figure}[H]
    \centering 
    \includegraphics[width=\textwidth]{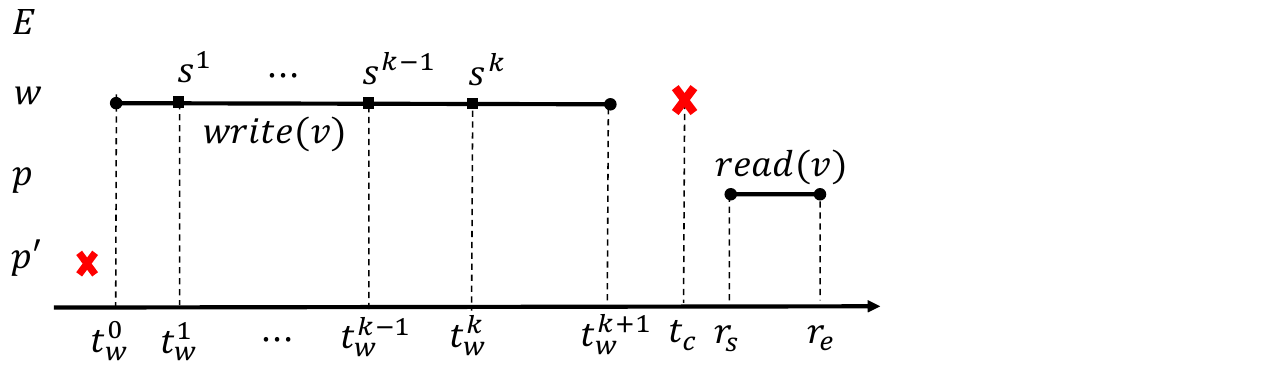}
    \caption{Execution $E$ (only the steps of $w$, $p$ and $p'$ are illustrated here)} 
    \label{e}
\end{figure}

\begin{figure}[H]
    \centering 
    \includegraphics[width=\textwidth]{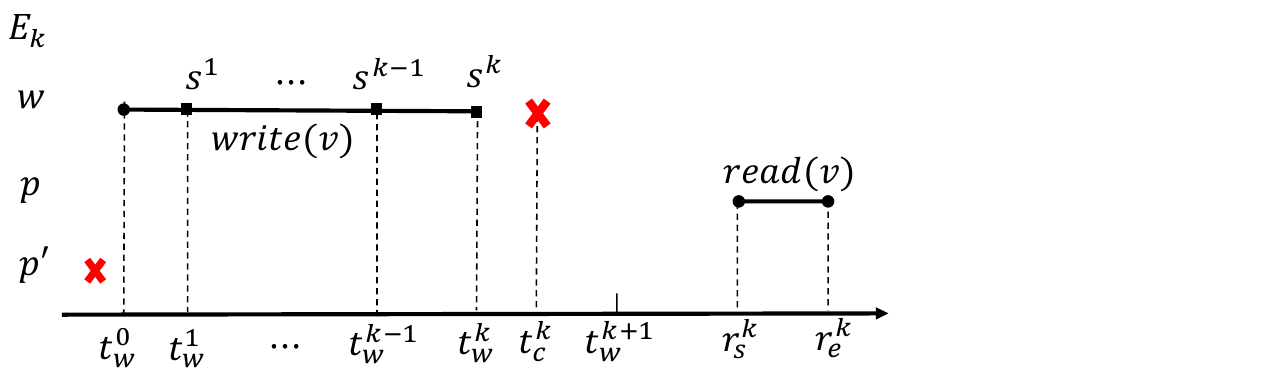}
    \caption{Execution $E_k$} 
    \label{ek}
\end{figure}


\begin{figure}[H]
    \centering 
    \includegraphics[width=\textwidth]{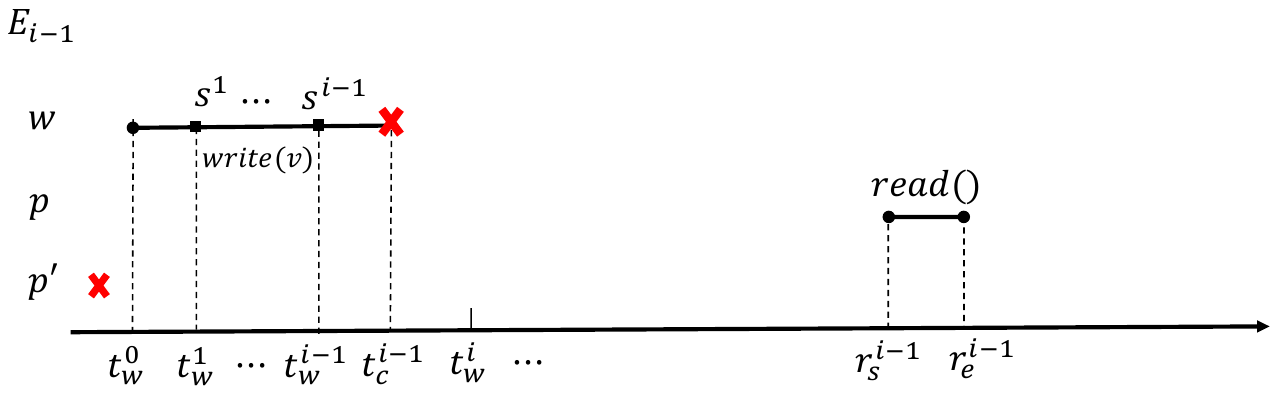}
    \caption{Execution $E_{i-1}$} 
    \label{ei-1}
\end{figure}

\begin{figure}[H]
    \centering 
    \includegraphics[width=\textwidth]{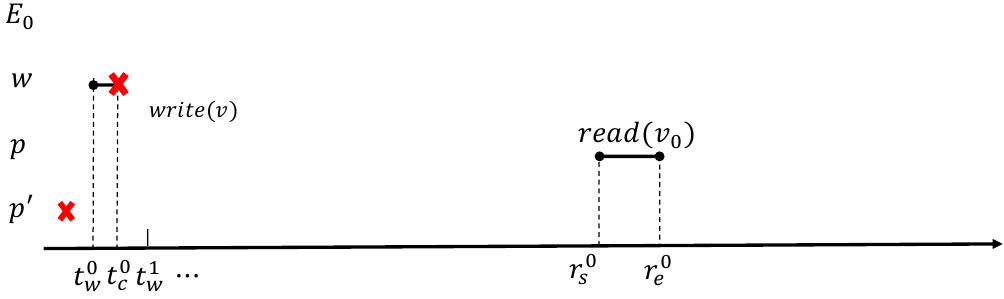}
    \caption{Execution $E_0$} 
    \label{e0}
\end{figure}

We now construct a sequence of executions $E_k, E_{k-1},\ldots, E_1, E_0$
	inductively as follows:
	
%

Execution $E_k$ of $\AW$ is identical to $E$ except that $w$ crashes
	at some time $t^k_c$, where $t^k_w < t^k_c <t^{k+1}_w$;
	that is,
	 $w$ crashes after executing all the $\pv$
	steps of $write(v)$, including $s^k$, but before the operation $write(v)$ returns (see Figure~\ref{ek}).

Since $w$ completes all the $\pv$
	steps of $write(v)$ before crashing, $E_k$ and $E$ are indistinguishable to all the processes in $P$, including $p$.
So $p$ behaves in $E_k$ as it did in $E$:
	at time $r^k_s = r_s$,
	process $p$ starts reading $\REG$,
	and at time $r^k_e = r_e$ this operation
	completes and returns $v$.

For $i \in \{1, \ldots, k\}$,
	$E_{i-1}$ is obtained from $E_{i}$ by making process $w$ crash
	one $\pv$ step earlier, i.e., just before executing step $s^{i}$ (see Figures~\ref{ei-1}). 
More precisely, $E_{i-1}$ is as follows: 

\begin{itemize}  

\item All the processes in $P'$ are crashed: they take no steps.

\item All the processes in $P$ are correct.



\item Process $w$ behaves exactly as in execution $E_{i}$ until it crashes at some time $t^{i-1}_c$, where  $t^{i-1}_w < t^{i-1}_c <t^{i}_w$,
	so $w$ crashes 
	\emph{before} executing $\pv$ step $s^{i}$.

\item All the processes in $P$ behave exactly as in execution $E_{i}$ up to and including time~$t^{i-1}_w$.

%
	
	
\item At some time $r^{i-1}_s >  t^{i-1} _c$,
	process $p$ starts reading $\REG$, and at time $r^{i-1}_e$ this operation
	completes and returns some $v_{i-1} \in \{v_0,v\}$.


\end{itemize}

Note that in execution $E_0$, process $w$ crashes at some time  $t^{0}_c$, where  $t^{0}_w < t^{0}_c <t^{1}_w$,
	before executing its first $\pv$ step $s^1$ (see Figure~\ref{e0}).
Since $w$ crashes before executing any communication step, processes in $P$ cannot distinguish execution $E_0$
	from one where
	$w$ crashes \emph{before starting any $write()$ operation}.
Thus, when $p$ reads $\REG$ in $E_0$, it reads the initial value of $\REG$, namely $v_0$.

\begin{claim}\label{biv}
There is an $i \in \{1, \ldots, k\}$ such that
	process $p$ reads $v_0$ from $\REG$ in $E_{i-1}$,
	and 
	process $p$ reads $v$ from $\REG$ in $E_{i}$.
\end{claim}

\begin{proof}
This is because $\forall i,  0 \le i \le k$,
	$p$ reads either $v_0$ or $v$ from $\REG$ in $E_{i}$,
	and $p$ reads $v_0$ from $\REG$ in $E_0$ and reads $v$ from $\REG$ in $E_{k}$.
\end{proof}

Henceforth, let $j \in \{1, \ldots, k\}$ be such that
	$p$ reads $v_0$ from $\REG$ in $E_{j-1}$
	and
	$p$ reads $v$ from $\REG$ in $E_{j}$
	(see Figures~\ref{ej} and~\ref{ej-1}).


\begin{claim}\label{toP}
The step $s^j$ of $w$ in execution $E_j$
	is one of the following two types:
	$w$ sends a message to a process in $P$,
	or $w$ writes a shared register that a process in $P$ can read.
\end{claim}

\begin{proof}
Note that:
(1)    the $\pv$ steps executed by $w$ in $E_{j-1}$ and $E_{j}$ differ only
	in that $w$ executes $s^j$ in $E_{j}$, but crashes before executing $s^j$ in $E_{j-1}$;
(2)	process $p$ is able to distinguish between $E_{j-1}$ and $E_{j}$
	(because $p$ reads $v_0$ from $\REG$ in $E_{j-1}$,
	while it reads $v$ from $\REG$ in $E_{j}$).
	
From the definition of $\pv$ steps, step $s^j$ of $w$ is one of the following:
	$w$ receives a set of messages,
	$w$ sends a message,
	$w$ reads a shared register, 
	or
	$w$ writes a shared register.
From (1) and (2), it is clear that $s^j$ cannot be a message receipt or a read step.
Furthermore, since all the processes in $P'$ take no steps (in both $E_{j-1}$ and $E_{j}$),
	$s^j$ must be either
	the sending of a message to \emph{a process in $P$}, or
	the writing of a register that can be read by \emph{a process in $P$}.
\end{proof}

\begin{figure}[H]
    \centering 
    \includegraphics[width=0.9\textwidth]{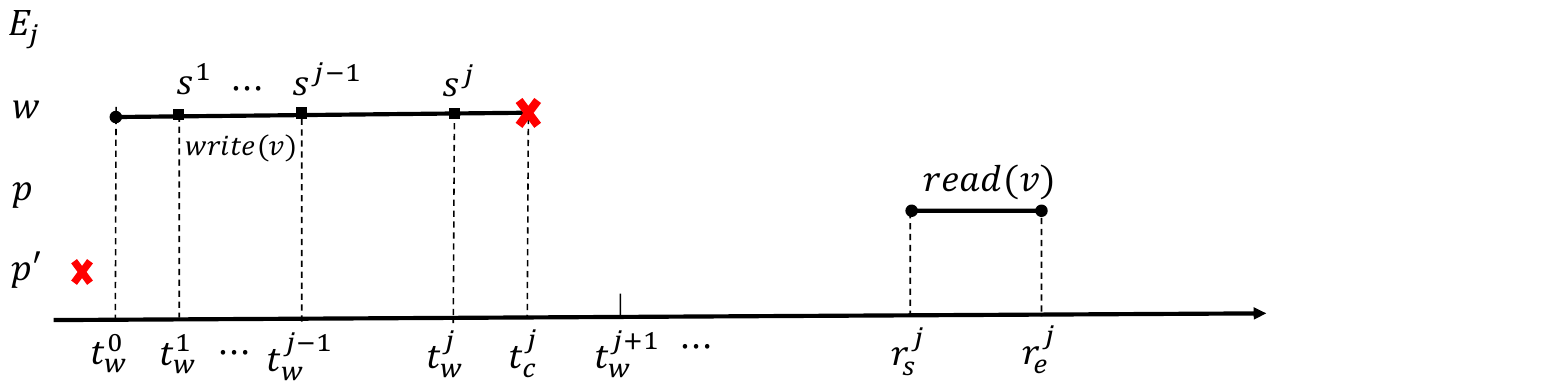}
    \caption{Execution $E_{j}$} 
    \label{ej}
\end{figure}

\vspace{-3mm}

\begin{figure}[H]
    \centering 
    \includegraphics[width=0.9\textwidth]{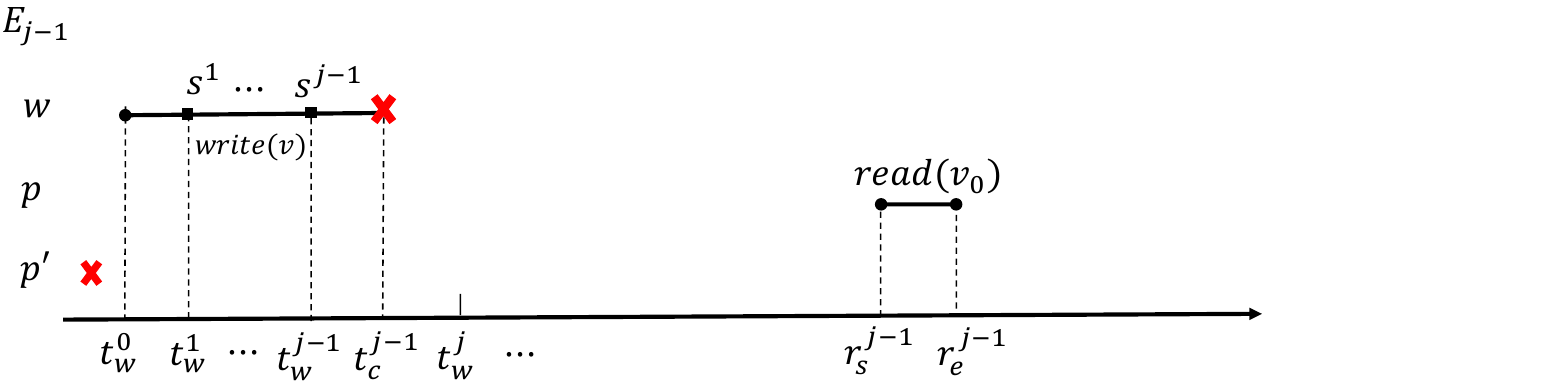}
    \caption{Execution $E_{j-1}$} 
    \label{ej-1}
\end{figure}

\vspace{-3mm}

\begin{figure}[H]
    \centering 
    \includegraphics[width=0.9\textwidth]{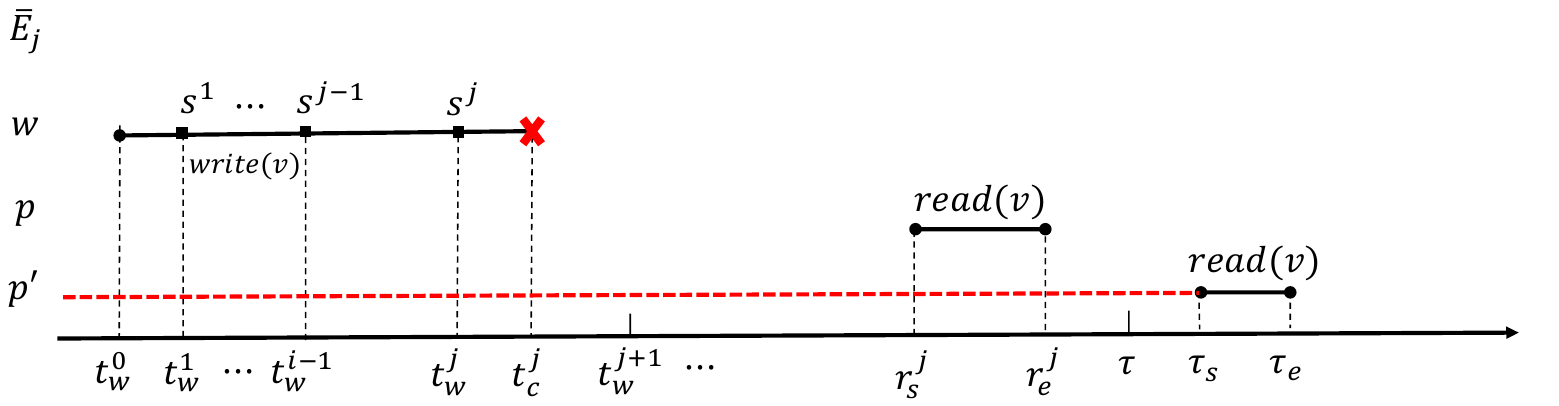}
    \caption{Execution $\bar{E}_j$} 
    \label{barej}
\end{figure}

\vspace{-3mm}

\begin{figure}[H]
    \centering 
    \includegraphics[width=0.9\textwidth]{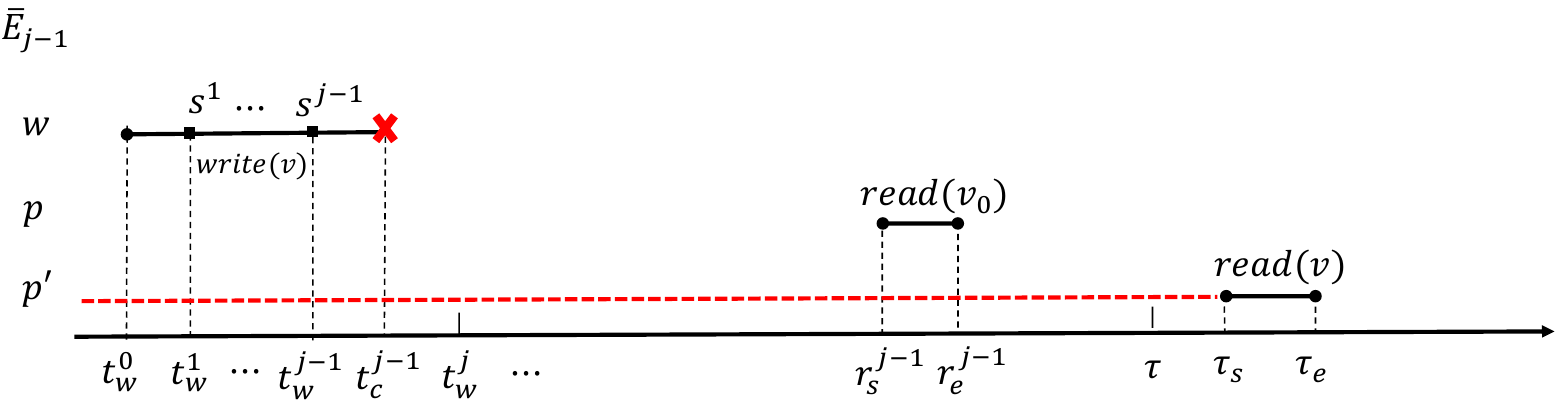}
    \caption{Execution $\bar{E}_{j-1}$} 
    \label{barej-1}
\end{figure}

\vspace{-3mm}

\begin{figure}[H]
    \centering 
    \includegraphics[width=0.9\textwidth]{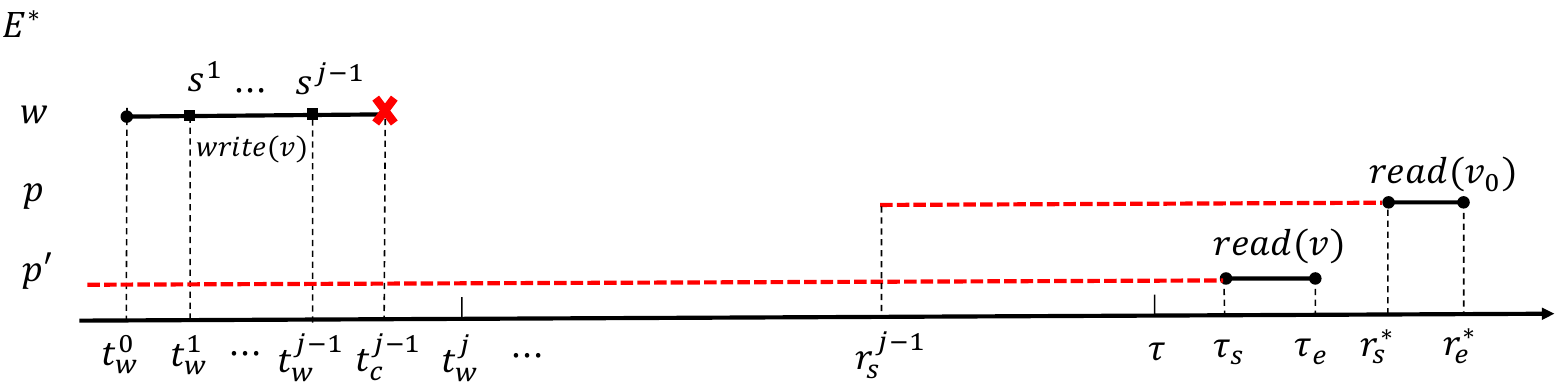}
    \caption{Execution $E^*$} 
    \label{estar}
\end{figure}

We now construct execution $\bar{E}_j$: roughly speaking, this execution
	is identical to $E_j$ except that all the processes in $P'$ (which were crashed in $E_j$) now ``wake up''
	after $p$ reads $v$ from $\REG$;
	and, after waking up, $p' \in P'$ reads $v$ from $\REG$ (see Figure~\ref{barej}). 
	
More precisely, $\bar{E}_j$ is as follows:

\begin{itemize}   

\item All processes behave exactly as in $E_j$ up to some time $\tau > \max( r^{j-1}_e, r^j_e)$.
	
		Recall that at time $r^{j-1}_e$, process $p$ completes its read of $v_0$ from $\REG$ in $E_{j-1}$;
		and that at time $r^{j}_e$, process $p$ completes its read of $v$ from $\REG$ in $E_{j}$.

	So, as in $E_j$:
	
	(i) Process $w$ crashes at time $t^j_c$, where  $t^j_w < t^j_c <t^{j+1}_w$.
	
	(ii) Process $p$ reads $v$ from $\REG$, and this operation starts at time $r^j_s$ and completes at time $r^j_e$.
	
	(iii) No process in $P'$ takes any step before time $\tau$.

\item At time $\tau$, all the processes in $P'$ start taking steps.

	After waking up, the processes in $P'$ receive all the messages that $w$
	and processes in $P$ sent to them up to and including time $t^{j-1}_w$; let $\cal{M}$ be this set of messages.
	But all the messages sent to them by $w$ and processes in $P$ \emph{after time $t^{j-1}_w$}
	are delayed until a time to be determined later.



	Recall that by (**), in system $\mathcal{S}_{L}$ no process in $P$ can write to a register that a process in $P'$ can read\MP{To generalize Theo: change to ``no process in $P$ can communicate with a process in $P'$ via a shared object"}.
	Thus, as long as we delay the receipt of the messages that processes in $P$ send to processes in $P'$
	after time $t^{j-1}_w$,
	for processes in $P'$ this execution is indistinguishable from one in which all the processes
	in $P$ and $Q$ have crashed by time $t^{j-1}_w$.
	Note that by ($\dagger$), 
	$|P \cup Q|=t$,
	so this is number of crashes is possible.

\item At some time $\tau_s$ after processes in $P'$ have received all the messages in $\cal{M}$,
	process $p'$ starts reading $\REG$.
	Since this is after $p$ completed its read of $v$ from $\REG$ at time $r^j_e$,
	$p'$ also reads $v$ from $\REG$.
Let $\tau_e$ be the time when $p'$ completes this read operation.	

\item After time $\tau_e$, all the processes in $P'$ receive all the delayed messages.

\end{itemize}


We now construct execution $\bar{E}_{j-1}$ which is obtained from $E_{j-1}$ in the same way that we obtained $\bar{E}_{j}$ from $E_{j}$:
	$\bar{E}_{j-1}$
	is identical to $E_{j-1}$ except that all the processes in $P'$ (which were crashed in $E_{j-1}$) now ``wake up''
	after $p$ reads $v_0$ from $\REG$;
	and after waking up, 
	$p' \in P'$ reads $\REG$ (see Figure~\ref{barej-1}). 
More precisely, $\bar{E}_{j-1}$ is as follows:

\begin{itemize}   

\item All processes behave exactly as in $E_{j-1}$ up to time $\tau$.
		Recall that $\tau > \max( r^{j-1}_e, r^j_e)$. 
	So, as in $E_{j-1}$:
	
	(i) Process $w$ crashes at time $t^{j-1}_c$, where $t^{j-1}_w < t^{j-1}_c <t^j_w$,
	so $w$ crashes ``just before'' executing step $s^{j}$.
	
        (ii) Process $p$ reads $v_0$ from $\REG$, and this operation starts at time $r^{j-1}_s$ and completes at time $r^{j-1}_e$.

	(iii) No process in $P'$ takes any step before time $\tau$.

\item At time $\tau$, all the processes in $P'$ start taking steps.

	After waking up, the processes in $P'$ receive all the messages that $w$
	and processes in $P$ sent to them up to and including time $t^{j-1}_w$; let $\cal{M}'$ be this set of messages.
	But all the messages sent to them by $w$ and processes in $P$ \emph{after time $t^{j-1}_w$}
	are delayed until a time to be determined later.
	
	Recall that $\cal{M}$ is the set of messages that $w$
	and processes in $P$ sent to processes in $P'$ up to and including time $t^{j-1}_w$ in execution $\bar{E}_{j}$.

%

\begin{claim}\label{samemsg}
	$\cal{M}' = \cal{M}$.
\end{claim}

\begin{proof}
All the messages that $w$ sends up to and including time $t^{j-1}_c$ are the same in $\bar{E}_{j-1}$  and $\bar{E}_{j}$. 
Furthermore, by the construction of $\bar{E}_{j-1}$ from $E_{j-1}$ and of $\bar{E}_{j}$ from $E_{j}$,
it is clear that
	up to and including time $t^{j-1}_w$,
 	all the processes in $P$ behave the same in $\bar{E}_{j-1}$ and $E_{j-1}$,
	and they also behave the same in $\bar{E}_{j}$ and $E_{j}$.
Furthermore,
 	by the construction of $E_{j-1}$ from $E_{j}$,
 	up to and including time $t^{j-1}_w$,
 	all the processes in $P$ behave the same in $E_{j-1}$ and $E_{j}$.
So up to and including time $t^{j-1}_w$
	all the processes in $P$ behave the same in $\bar{E}_{j-1}$ and $\bar{E}_{j}$.
Thus,
	all the messages that processes in $P$ send to processes in $P'$
	up to and including time $t^{j-1}_w$ are the same in $\bar{E}_{j-1}$ and $\bar{E}_{j}$. 

From the above, 
	and the definition of $\cal{M}'$ and $\cal{M}$,
	it is now clear that $\cal{M}' = \cal{M}$.
\end{proof}

	Recall that in system $\mathcal{S}_{L}$, no process in $P$ can write to a register
	that a process in $P'$ can read\MP{To generalize Theo: change to ``no process in $P$ can communicate with a process in $P'$ via a shared object"}.
	Moreover, by Claim~\ref{toP},
	step $s^j$ is not a write to a register that any process in $P'$ can read
	(because $P \cap P' = \emptyset$ and no set $
	S_i$ contains a node in $P$ and a node in $P'$).
	Thus,
	from Claim~\ref{samemsg},
		as long as we delay the receipt of the messages that processes in $P$ send to processes in $P'$ after time
		$t^{j-1}_w$,
	for processes in $P'$,
	this execution is indistinguishable from $\bar{E}_j$.

\item At time $\tau_s$ after processes in $P'$ have received all the messages in $\cal{M}'=\cal{M}$,
	process $p'$ starts reading $\REG$.
Since for processes in $P'$ this execution is indistinguishable from $\bar{E}_j$ (so far),
	$p'$ reads $v$ from $\REG$ as in $\bar{E}_j$,
	and this read operation completes at time $\tau_e$ as in $\bar{E}_j$.
	
\item After time $\tau_e$, all the processes in $P'$ receive all the delayed messages.
\end{itemize}

Finally, we construct the execution $E^*$ that yields the contradiction.
Roughly speaking,
	$E^*$ is obtained from $\bar{E}_{j-1}$ by delaying the read operation of $p$: 
	in $\bar{E}_{j-1}$, the read operation of $p$ \emph{completes before} the read of $p'$ starts,
	while in $E^*$, the read operation of $p$ \emph{starts after} the read of $p'$ completes (see Figure~\ref{estar}).
More precisely, $E^*$ is as follows:

\begin{itemize}   

\item Process $w$ behaves exactly as in execution $\bar{E}_{j-1}$.

	
\item Processes in $P$ behaves the same as in $\bar{E}_{j-1}$, up to but \emph{not} including time $r^{j-1}_s$;
	at time $r^{j-1}_s$ they pause (we will see later when they will resume taking steps).
	In particular, process $p$ does \emph{not} invoke the read of $\REG$ at time $r^{j-1}_s$.
	
\item Every process in $P'$ behaves exactly as in $\bar{E}_{j-1}$ up to and including time $\tau_e$. In particular:

	(1) No process in $P'$ takes any step before time $\tau$.
	
	(2) At time $\tau$, all the processes in $P'$ start taking steps.

	After waking up, the processes in $P'$ receive all the messages that $w$
	and processes in $P$ sent to them up to and including time $t^{j-1}_w$
	(i.e., they receive all the messages in $\cal{M'}=\cal{M}$).
	
	(3) After processes in $P'$ have received all these messages,
	at time $\tau_s$ process $p'$ starts reading $\REG$,
	$p'$ reads $v$ from $\REG$,
	and this read operation completes at time $\tau_e$.

	Note that each process in $P'$ behaves exactly as in $\bar{E}_{j-1}$ up to and including time $\tau_e$ since it cannot ``notice''
	that in $E^*$ (in contrast to $\bar{E}_{j-1}$) processes in $P$ paused from time $r^{j-1}_s$.
	This is because:
	(a) in $\bar{E}_{j-1}$,
	all the messages that $w$ and processes in $P$ send to processes in $P'$
	\emph{after time $t^{j-1}_w$} 
	are delayed and received \emph{after} time $\tau_e$, and
	(b) by (**),
		 in system $\mathcal{S}_{L}$, no process in $P$ can write to a register
	that a process in $P'$ can read.\MP{To generalize Theo: change to ``no process in $P$ can communicate with a process in $P'$ via a shared object"}

\item All the messages that processes in $P'$ send after they wake up at time $\tau$ are delayed until a time to be determined later.
	
\item After $p'$ completes reading $v$,
	at some time $r_s^* > \tau_e$
	all the processes in $P$ resume taking steps,
	and the steps that they take from time $r_s^*$
	are exactly the same as those that they take in $\bar{E}_{j-1}$ from time $r^{j-1}_s$:
	so these steps are just delayed by $\delta = r_s^* - r^{j-1}_s$.
	Intuitively, processes in $P$ do not ``notice'' that this delay occurred.
	More precisely, processes in $P$ cannot distinguish between
	$E^*$ and $\bar{E}_{j-1}$ up to and including time $r_e^* = r^{j-1}_e + \delta$ ($\ddag$).
	To see why, note that:
	(a) up to and including time $r_e^*$,
		processes in $P$ do not receive any message from processes in $P'$, exactly as
		in execution $\bar{E}_{j-1}$ up to and including time $r^{j-1}_e$;
	and 
	(b) in system $\mathcal{S}_{L}$, no process in $P'$ can write to a register
	that a process in $P$ can read.\MP{To generalize Theo: change to ``no process in $P'$ can communicate with a process in $P$ via a shared object"}

	By ($\ddag$), $p \in P$ starts reading $\REG$ at time $r_s^* = r^{j-1}_s + \delta$ (just as it did at time $r^{j-1}_s$ in $\bar{E}_{j-1}$),
	it reads $v_0$ and completes this read at time $r_e^* = r^{j-1}_e + \delta$
	(just as it did at time $r^{j-1}_e$ in $\bar{E}_{j-1}$).

\item After time $r_e^*$, all the processes in $P$ and $P'$ receive all the delayed messages.
\end{itemize}

Note that in execution $E^*$,
	process $p$ reads (the ``old'' value) $v_0$ from $\REG$ but this read starts at time $r_s^*$ \emph{after} the time $\tau_e$ when process $p'$ completes reading (the ``new'' value) $v$ from $\REG$.
This new-old inversion violates the atomicity of register $R$ --- a contradiction.

\end{proof}

Note that the proof of Theorem~\ref{thm-lb} does not depend on the \emph{type} or \emph{number} of registers
	shared by the processes in each set $S_i$ of the bag $L$.
In fact, the proof of Theorem~\ref{thm-lb} does not even depend on the type of objects
	that are shared by the processes in each set $S_i$; for example these objects could include
	queues, stacks, and consensus objects.
So the result of this theorem applies not only to $\mathcal{S}_{L}$ but also to every 
	m\&m system $\mathcal{S}$ in $\UM_{L}$
	where the processes in each $S_i$ share any number of objects of any type among themselves (and only among themselves).
Hence we have the following stronger result:

\begin{theorem}\label{thm-lb-on-steroid}
Consider any 
	m\&m system $\mathcal{S}$ in $\UM_{L}$ induced
	by a bag $L =  \{S_1,\dots ,S_m\}$ of subsets of $\Pi = \{ p_1 , p_2 , \ldots, p_n \}$.
If $t_L+1 <n$ processes may crash in $\mathcal{S}$,
	then
	for every process $w$ in $\mathcal{S}$,
	there is no algorithm that implements an atomic SWMR 
	register writable by $w$ and readable by all processes in $\mathcal{S}$.
\end{theorem}

\section{Atomic SWMR register implementation in uniform m\&m systems}\label{noname}

Let $G=(V,E)$ be an undirected graph where $V = \{ p_1 , p_2 , \ldots, p_n \} $,
	 i.e., the nodes of $G$ are the processes $p_1 , p_2 , \ldots, p_n$.
Let $\mathcal{S}_{G}$ be the uniform m\&m system induced by $G$.
Recall that in $\mathcal{S}_{G}$,
	each process $p_i$ and its neighbours in $G$
	share some atomic SWMR registers that can be read by (and only by) them.

We now use $G$ to determine the maximum number of process
	crashes that may occur in $\mathcal{S}_{G}$ such that
	it is possible to implement a shared atomic SWMR register readable by all processes in $\mathcal{S}_G$.
To do so, we first recall the definition of the \emph{square of the graph $G$}: 
$G^2 = (V,E')$ where $\mbox{$E' = \{ (u,v) ~|~ (u,v) \in E$}$ or  $\mbox{$\exists k \in V$} \mbox{ such that }
	(u,k) \in E \mbox{ and } (k,v) \in E \}$.

\begin{definition}\label{doesnothaveone}
Given an undirected graph $G=(V,E)$ such that $V = \{ p_1 , p_2 , \ldots, p_n \}$,
	$t_G$ is the maximum integer $t$ such that the following condition holds:
For all disjoint subsets $P$ and $P'$ of $V$ of size $n-t$ each,
	some edge in $G^2$ connects a node in $P$ with a node in $P'$; i.e.,
	$G^2$ has an edge $(u,v)$ such that $u \in P$ and $v \in P'$.
\end{definition}

\begin{figure}
\centering
\begin{minipage}{.5\textwidth}
  \centering
 \hbox{\hspace{1cm} \includegraphics[width=.60\linewidth]{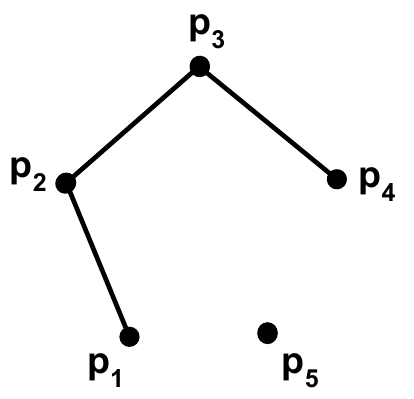}}
  \caption{A graph $G$}
  \label{G-fig2}
\end{minipage}%
\begin{minipage}{.5\textwidth}
  \centering
 \hbox{\hspace{1cm}\includegraphics[width=.6\linewidth]{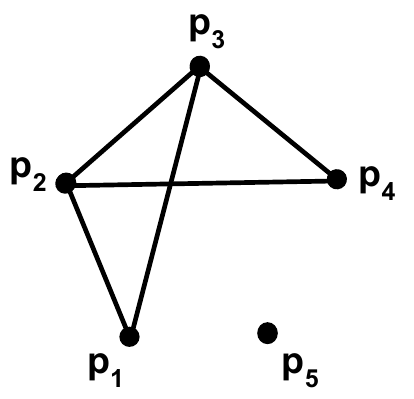}}
  \caption{The square $G^2$ of graph $G$}
  \label{G2-fig}
\end{minipage}
\end{figure}

Note that $\lceil n/2\rceil -1 \le t_G \le n-1$.
Moreover,
	in a pure message-passing system (where $G$ and $G^2$ have no edges)
	$t_G = \lceil n/2\rceil -1$.
	
To illustrate the above definition of $t_G$, 
	consider the graph $G$ in Figure \ref{G-fig2} where $V= \{ p_1, p_2, p_3, p_4, p_5 \}$.
Figure~\ref{G2-fig} shows the corresponding $G^2$ graph.
By the above definition of $t_G$:
	(a) $t_G \ge  3$ because for any two disjoint subsets of $V$ of size $5-3=2$ each,
	$G^2$ has an edge that ``connects'' these two subsets
	(e.g., for subsets $P= \{ p_1, p_2 \}$ and $P' = \{p_4, p_5 \}$,
	the edge $(p_2, p_4)$ of $G^2$ connects a node of $P$ to a node of $P'$), and
	(b) $t_ G< 4$ because there are two disjoint subsets $\{ p_1 \}, \{p_5 \}$ of size $5-4=1$ each,
	 such that no edge in $G^2$ connects $p_1$ and $p_5$.
So in this example $n = 5$ and $t_G = 3$. 

In Theorem~\ref{thm44} given below, we show that in the uniform m\&m system $\mathcal{S}_G$ induced by a graph~$G$, it is possible to implement an atomic SWMR register readable by all processes \emph{if~and~only~if} at most $t_G$ processes may crash in $\mathcal{S}_G$.

For example, consider the uniform m\&m system $\mathcal{S}_G$ of 5 processes induced by the graph $G$ in Figure \ref{G-fig2}.
In addition to message-passing links, $\mathcal{S}_G$ has 3 pairwise RDMA connections.
Since $t_G =3$, by Theorem~\ref{thm44},
	we can implement an atomic SWMR register readable by all
	5 processes of $\mathcal{S}_G$ if and only if at most 3 of them may crash.
In contrast, in a pure message-passing system with 5 processes, no implementation of such a register
	can tolerate more than 2 process crashes.

To prove Theorem~\ref{thm44} we first show:

\begin{lemma}\label{relatetgandtl}
Let $\mathcal{S}_{G}$ be the uniform m\&m system induced by
	an undirected graph \linebreak $G=(V,E)$ where $V = \{ p_1 , p_2 , \ldots, p_n \}$.
Let $\mathcal{S}_{L}$ be the general m\&m system
	such that $\mathcal{S}_{L} = \mathcal{S}_{G}$.
Then $t_G = t_L$.
\end{lemma}

\begin{proof}

By Definition~\ref{SG-def},  $\mathcal{S}_{L}$ is the general m\&m system 
	where~$L =  \{  S_1, S_2,\dots, S_n \}$ such that
	$S_i = N^+(p_i)$, i.e., for all $i$, $1 \le i \le n$,
	$S_i$ is the set of neighbours of $p_i$ (including $p_i$) in the graph $G$.
Recall that $t_L$ is the maximum $t$ such that for all disjoint subsets $P$ and $P'$ of $V$ of size $n-t$ each,
	some set $S_i$ in $L$ contains both a node in $P$ and a node in $P'$.

From the definitions of $t_G$ (Definition~\ref{doesnothaveone}) and $t_L$, it is clear that to prove that $t_G = t_L$
	it suffices to show that for all $0 \le t \le n$ and all disjoint subsets $P$ and $P'$ of $V$ of size $n-t$ each, the following holds:
	some edge in $G^2$ connects a node in $P$ with a node in $P'$
	if and only if some set $S_i$ in $L$ contains both a node in $P$ and a node in $P'$.

\noindent
[\textsc{Only If}] Suppose $G^2$ has an edge $(p_i , p_j)$ such that $p_i \in P$ and $p_j \in P'$; since $P$ and $P'$ are disjoint, $p_i$ and $p_j$ are distinct.
By definition of $G^2$, there are two cases: 
\begin{enumerate}
\item $(p_i , p_j) \in E$. 
In this case, $p_j \in N^+(p_i)$ and $p_i \in N^+(p_i)$. 
So the set $S_i = N^+(p_i)$ in $L$ contains both node $p_i \in P$ and node $p_j \in P'$.

\item There is a node $p_k \in V$ such that  $(p_i , p_k) \in E$ and $(p_k , p_j) \in E$.
 In this case, $p_i \in N^+(p_k)$ and $p_j \in N^+(p_k)$.
So the set $S_k = N^+(p_k)$ in $L$ contains both $p_i \in P$ and $p_j \in P'$.
\end{enumerate}
So in both cases, some set $S_{\ell}$ in $L$ contains both a node in $P$ and a node in $P'$.

\noindent
[\textsc{If}] Suppose set $S_k$ in $L$ contains both a node $p_i$ in $P$ and a node $p_j$ in $P'$; since $P$ and $P'$ are disjoint, $p_i$ and $p_j$ are distinct.
Recall that $S_k=N^+(p_k)$ for node $p_k \in V$.

There are two cases:
\begin{enumerate}
\item $p_i$, $p_j$ and $p_k$ are pairwise distinct. 
In this case, since $p_i$ and $p_j$ are in $S_k = N^+(p_k)$, $(p_i,p_k)$ and $(p_k,p_j)$ are edges of $G$, i.e., $(p_i,p_k) \in E$ and $(p_k,p_j) \in E$.
Thus, by definition of $G^2$, $(p_i, p_j)$ is an edge of $G^2$; this edge connects $p_i \in P$ and $p_j \in P'$.

\item $p_k = p_i$ or $p_k = p_j$.
Without loss of generality, assume that $p_k = p_i$. 
Since $p_i$ and $p_j$ are in $N^+(p_k) = N^+(p_i)$, $(p_i,p_j)$ must be an edge of $G$, i.e., $(p_i,p_j) \in E$.
Thus, by definition of $G^2$, $(p_i, p_j)$ is an edge of $G^2$; this edge connects $p_i \in P$ and $p_j \in P'$.
\end{enumerate} 

So in both cases, some edge in $G^2$ connects a node in $P$ with a node in $P'$.
\end{proof}

From Lemma~\ref{relatetgandtl} and Theorem~\ref{thm0}, we have:

\begin{theorem}\label{thm44}
Let $\mathcal{S}_{G}$ be the uniform m\&m system induced by
	an undirected graph \linebreak $G=(V,E)$ where $V = \{ p_1 , p_2 , \ldots, p_n \}$.

\begin{itemize}

\item If at most $t_G$ processes may crash in $\mathcal{S}_{G}$,
	then
	for every process $w$ in~$\mathcal{S}_{G}$,
	it is possible to implement an atomic SWMR register writable by $w$ and
	readable by all processes in~$\mathcal{S}_{G}$.

\item If $t_G +1 <n$ processes may crash in $\mathcal{S}_{G}$,
	then
	for every process $w$ in $\mathcal{S}_{G}$,
	it is impossible to implement an atomic SWMR register writable by $w$ and
	readable by all processes in~$\mathcal{S}_{G}$.

\end{itemize}
\end{theorem}

\begin{figure}
\begin{minipage}{.5\textwidth}
 \centering
 \hbox{\hspace{1cm}\includegraphics[width=.6\linewidth]{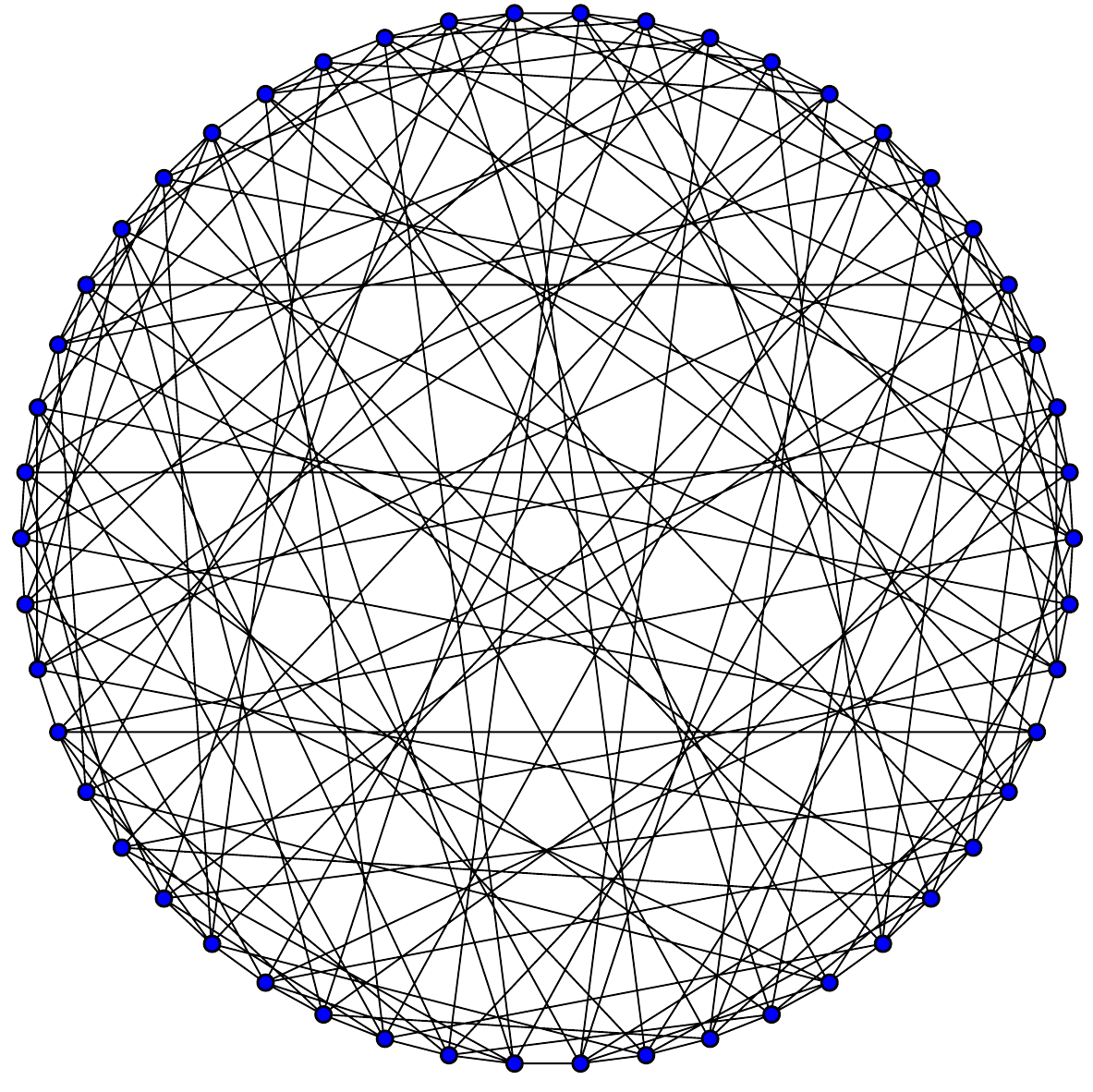}}
   \caption{The Hoffman-Singleton graph}
  \label{G3-fig}
\end{minipage}%
\begin{minipage}{.5\textwidth}
 \centering
 \hbox{\hspace{1cm}\includegraphics[width=.6\linewidth]{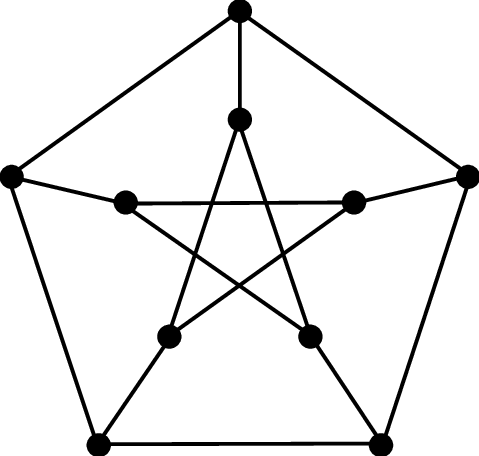}}
  \caption{The Petersen Graph}
  \label{G4-fig}
\end{minipage}
\end{figure}

To illustrate this theorem, we now give three examples.
For our first example, consider a pure message-passing system $\mathcal{S}$ with 50 nodes.
In~$\mathcal{S}$, one can implement an atomic SWMR register $\REG$ (readable by all the processes)
	only if \emph{at most} 24 process crashes may occur.
But if we allow each process of $\mathcal{S}$ to establish 7 pairwise
	 RDMA connections,
	one can implement $\REG$ in a way that tolerates \emph{any number} of process crashes
	(i.e., $\REG$ is wait-free).
This is because there is an undirected graph $G$ with $n=50$ nodes, each with degree 7,
	such that $G^2$ has an edge between \emph{every} pair of nodes
	($G$ is the well-known \emph{Hoffman-Singleton} graph~\cite{HoffmanSingleton} shown in Figure~\ref{G3-fig}~\cite{HS-Fig});
	so $t_G = n-1 =49$, and thus by Theorem~\ref{thm44}
	it is possible to implement $\REG$ in the uniform m\&m system $\mathcal{S}_{G}$
	in a way that tolerates up to 49 process crashes.
Some simple graph theory arguments show that this is optimal in two ways:
	(a) one cannot implement a wait-free register $\REG$ that is shared by 50 processes
	with fewer than 7 RDMA connections per process (more precisely, with any such implementation,
	if a process has fewer than 7 RDMA connections there must be another process with more than 7 RDMA connections),
	and (b) with at most 7 RDMA connections per process,
	one cannot implement a wait-free register $\REG$ that is shared by more than 50 processes.

As another example, consider the uniform m\&m system $\mathcal{S}_{G}$
	with $n=10$ processes and 3 RDMA connections per process induced by \emph{Petersen graph} $G$ shown in Figure~\ref{G4-fig}.\footnote{As with the Hoffman-Singleton graph,
	 Petersen graph is a \emph{Moore Graph} with diameter 2~\cite{HoffmanSingleton}.}
Since $G$ has diameter 2, $G^2$ has an edge between every pair of nodes,
	and so $t_G = n-1 = 9$.
Thus, by Theorem~\ref{thm44}, 
	one can implement an atomic SWMR register $\REG$
	in $\mathcal{S}_{G}$ in a way that tolerates up to 9 process crashes.
In contrast,  in a pure message-passing system
	with 10 processes, no implementation of $\REG$ can tolerate more than 4 process crashes.

As our last example, we show that \emph{expander graphs}
	with high \emph{vertex expansion ratio}~\cite{expandergraphs}
	induce uniform m\&m systems
	that support highly fault-tolerant register implementations.
To do so, first recall the definition of the vertex expansion ratio:

\begin{definition}
Let $G = (V, E)$ be any undirected graph.
\begin{enumerate}
\item The vertex boundary of a subset $S \subseteq V$ is 
$$\delta S = \{ v\in V: (u,v)\in E, u\in S,v \notin S\}$$
\item The vertex expansion ratio of $G$, denoted $h(G)$, is defined as 
$$h(G) = \min_{S \subseteq V: 0< |S|\leq |V|/2}\frac{|\delta S|}{|S|}$$
\end{enumerate}
\end{definition}

\noindent
We now prove that any graph $G$ with high vertex expansion ratio $h$ also has a large $t_G$.

\begin{theorem}\label{exp}
For any undirected graph $G$ with $n$ nodes and vertex expansion ratio $h$,
	$t_G \geq \lceil (1-\frac{1}{h^2+2h+2})n \rceil-1$.
\end{theorem}

\begin{proof}
Let $G=(V,E)$ be an undirected graph with $n$ nodes and vertex expansion ratio $h$,
To show $t_G \geq \lceil (1-\frac{1}{h^2+2h+2})n \rceil-1$,
	we must show that for every $t$, $0 \le t \le \lceil (1-\frac{1}{h^2+2h+2})n \rceil-1$, the following holds.
	For all disjoint subsets $P$ and $P'$ of $V$ of size $n-t$ each:
(*) some edge in $G^2$ connects a node in $P$ to 
	a node in~$P'$.

Let $t$ be such that $0 \le t \le \lceil (1-\frac{1}{h^2+2h+2})n \rceil-1$.
Clearly, $0 \le t < (1-\frac{1}{h^2+2h+2})n$.
Let $P$ and $P'$ be any two disjoint subsets of $V$ of size $m = n-t$ each.
We now show that~(*)~holds.

There are three cases:
	(1) $|P\cup\delta P|\leq n/2$,
	(2) $|P'\cup\delta P'|\leq n/2$, or
	(3) $|P\cup\delta P| > n/2$ and $|P'\cup\delta P'| > n/2$.

\textbf{Case 1}:  $|P\cup\delta P|\leq n/2$.
Since $|P| = m \leq |P\cup\delta P| \leq n/2$,
	by the definition of vertex expansion ratio $h$,
	$h \le |\delta P| / |P|$.
	Since $|P| = m$, we have
	$(h+1)m \leq |P\cup\delta P|$. 
Thus,
	again by the definition of vertex expansion ratio $h$,
	 $(h+1)^2m \leq |P \cup\delta P \cup \delta (P \cup\delta P)|$.

\begin{align*}
\text{By assumption: } ~~~
 t & < (1 - \frac{1}{h^2+2h+2})n\\
\Rightarrow 
\frac{n}{h^2+2h+2}  &< n-t \\
\Rightarrow 
\frac{n}{h^2+2h+2} &<m\\
\Rightarrow  \hspace{1.65cm}
n&< (h^2+2h+2)m\\
\Rightarrow \hspace{1.65cm}
n&<(h+1)^2m + m   
\end{align*}
Since $|P'| = m$, $|P \cup\delta P \cup \delta (P \cup\delta P)| \geq (h+1)^2m $, and $(h+1)^2m + m >n$, 
	the sets $P'$ and $P \cup\delta P \cup \delta (P \cup\delta P)$
  	intersect.
Thus,
	since $P'$ and $P$ are disjoint, 
	there is a node $q$ in $P'$ such that:
	either (i) $q$ is in $\delta P$, so it is connected to a node in $P$ by an edge in $G$, or
	(ii) $q$ is in $\delta (P\cup\delta P)$, so it is connected to a node in $P$ by a two-edge path in $G$.
Thus,
	by the definition of $G^2$,
	(*) holds

\textbf{Case 2}: $|P'\cup\delta P'|\leq n/2$.
	By a symmetric argument to Case 1,
	(*) holds.
	
\textbf{Case 3}: $|P\cup\delta P| > n/2$ and $|P'\cup\delta P'| > n/2$.
Then the sets $P\cup\delta P$ and $P'\cup\delta P'$ intersect.  
Thus,
	since $P$ and $P'$ are disjoint, 
	there are three cases: (1) $P$ and $\delta P'$ intersect, 
	so an edge in $G$ connects a node in $P$ to a node in $P'$,
	or (2) $P'$ and $ \delta P$ intersect, 
	so an edge in $G$ connects a node in $P'$ to a node in $P$,
	or (3) $\delta P'$ and $ \delta P$ intersect,
	so there are nodes $p\in P$ and $p' \in P'$ that are connected by a two-edge path in $G$.
Thus,
	in all cases,
	by the definition of $G^2$,
	(*) holds.
	
Therefore, in all cases, (*) holds.
\end{proof}

By Theorems~\ref{thm44} and~\ref{exp}, we have:

\begin{corollary}\label{corollaryTheo22}
Let $G$ be any undirected graph with $n$ nodes and vertex expansion ratio $h$.
If at most $\lceil (1-\frac{1}{h^2+2h+2})n \rceil-1$ processes crash in $\mathcal{S}_{G}$,
	then
	for every process $w$ in~$\mathcal{S}_{G}$,
	it is possible to implement an atomic SWMR register writable by $w$ and
	readable by all processes in~$\mathcal{S}_{G}$.
\end{corollary}

\section{Optimal randomized consensus in m\&m systems}

In the \emph{consensus problem}, each process \emph{proposes} some value 
	and must \emph{decide} a value such that the following properties hold:
	
\begin{itemize}
	\item \textbf{Validity}: Each decision value is one of the proposal values.
	\item \textbf{Agreement}: No two processes decide different values.
	\item \textbf{Termination}: Every process that does not crash eventually decides a value.
\end{itemize}

This problem cannot be solved
	in asynchronous distributed systems either with message-passing~\cite{flp},
	or with shared registers~\cite{LouiAbuAmara1987},
	but there are \emph{randomized} algorithms that solve
	\emph{randomized consensus}, a weaker version of the consensus problem that
		requires Termination ``only'' with probability 1.
In particular, in shared-memory systems, it is known that randomized consensus can be solved for any number of process crashes,
	but in message-passing systems, it can be solved if and only if fewer than half of the processes may crash.

We now show how to solve randomized consensus in \emph{m\&m systems} with the maximum fault-tolerance possible.
To do so, we combine the randomized consensus algorithm
	by Aspnes and Herlihy~\cite{AH1990consensus} (henceforth the AH algorithm),
	which was designed for shared-memory systems with \emph{atomic} SWMR registers,
	with the \emph{linearizable} implementation of such registers for m\&m systems that we gave in Section~\ref{algo}.
Doing so, however, is not as straightforward as it may seem:
	as pointed out in~\cite{sl1}, 
	a randomized algorithm that works with atomic registers
	does not necessarily work against a strong adversary
	if we replace the atomic registers that it uses
	with linearizable implementations of these registers.
In fact, it was shown that such a substitution may kill
	the termination property of a randomized algorithm~\cite{sltermination}.

In Section~\ref{consalgo}, we explain why we can indeed solve randomized consensus in m\&m systems
	by replacing the atomic registers used by the AH algorithm,
	with the linearizable implementation of atomic registers
	given in Section~\ref{algo}.
In Section~\ref{conslb}, we note that doing so is optimal in the number
	of processes crashes that it can tolerate in (both general and uniform) m\&m systems.
These results are summarized~by:

\begin{theorem}\label{thmcons0}
~
\begin{itemize}

\item Let $\mathcal{S}_{L}$ be the general m\&m system induced by
	a bag $L = \{S_1,\dots ,S_m\}$ of subsets of \linebreak  $\Pi = \{ p_1 , p_2 , \ldots, p_n \}$.
If at most $t_L$ processes may crash, randomized consensus can be solved;
	if $t_L+1 <n$ processes may crash, it cannot be solved.

\item Let $\mathcal{S}_{G}$ be the uniform m\&m system induced by
	an undirected graph $G=(V,E)$ where \linebreak $V = \{ p_1 , p_2 , \ldots, p_n \}$.
If at most $t_G$ processes may crash, randomized consensus can be solved;
	if $t_G+1 <n$ processes may crash, it cannot be solved.
\end{itemize}

\end{theorem}

The above theorem follows directly from Theorem~\ref{thmconsup} (Section~\ref{consalgo}),
	Theorem~\ref{thmconsdown} (Section~\ref{conslb}), and Lemma~\ref{relatetgandtl}.

\begin{figure}
\centering
\includegraphics[width=0.25\textwidth]{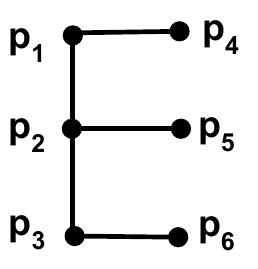}
  \caption{A graph $G$}
  \label{Gc}
\end{figure}

It is worth noting that the (optimal) fault-tolerance achieved
	by our randomized consensus algorithm for uniform m\&m systems
	is better than the fault-tolerance of the
	algorithm given for such systems in~\cite{aguilera-etal18}.\footnote{\cite{aguilera-etal18}
	considers randomized consensus algorithms only for uniform m\&m systems.}
For example, consider the undirected graph $G$ in Figure~\ref{Gc} and the corresponding m\&m system~$\mathcal{S}_G$.
It turns out that the randomized consensus algorithm given in~\cite{aguilera-etal18}
	tolerates at most 3 process crashes in system $\mathcal{S}_G$,
	but our algorithm tolerates up to $4$ process crashes in this system (because $t_G =4$ for this~$G$).

As another example, consider the Hoffman-Singleton graph $G$ (Section~\ref{noname}, Figure~\ref{G3-fig})
	and the corresponding m\&m system $\mathcal{S}_G$ with 50 processes.
As we explained in Section~\ref{noname}, our randomized consensus algorithm is wait-free,
	i.e., it tolerates up to $t_G = 49$ process crashes in $\mathcal{S}_G$.
In contrast, it can be shown that the algorithm given in~\cite{aguilera-etal18}
	cannot tolerate more than 45 process crashes in $\mathcal{S}_G$.

As a final example, consider any graph $G$ with $n$ nodes and expansion ration $h$.
The randomized consensus algorithm in \cite{aguilera-etal18} can tolerate
	at least $\lceil (1-\frac{1}{2h+2})n\rceil -1$ process crashes in the m\&m system $\mathcal{S}_G$ (Theorem 4.3 in~\cite{aguilera-etal18}).
In contrast,
	by Theorems~\ref{thmcons0} and Theorem~\ref{exp} we have:

\begin{corollary}\label{expansion}
For any undirected graph $G$ with $n$ nodes and vertex expansion ratio $h$,
	there is a randomized consensus algorithm that tolerates at least $\lceil(1-\frac{1}{h^2+2h+2})n\rceil -1$ process crashes in
	the uniform m\&m system $\mathcal{S}_G$.
\end{corollary}

\subsection{Solving randomized consensus}\label{consalgo}
The randomized consensus algorithm by Aspnes and Herlihy~\cite{AH1990consensus}
	was originally proved to work against a strong adversary
	in a shared-memory system
	under the assumption that the SWMR registers that it uses are atomic (i.e., instantaneous).
As we mentioned before, replacing the atomic registers of a randomized consensus
	algorithm with \emph{linearizable implementations} of atomic registers
	may kill the (probabilistic) termination property of this algorithm against a strong adversary~\cite{sltermination}:
	to preserve termination may require \emph{strongly linearizable} implementations,
	rather than just \emph{linearizable} implementations~\cite{sl1}.
As we show in Appendix~\ref{nsl}, however,
	our atomic register implementation for m\&m systems
	is \emph{not} strongly linearizable.\footnote{This also applies to the ABD register implementation for message-passing systems.}
 So \emph{prima facie}, 
	combining the AH algorithm with our implementations
	of atomic registers may not work against a strong adversary in m\&m systems.
	
It was recently shown~\cite{hadzilacos2020randomized},  however,
	that if we replace the atomic SWMR registers used by the AH algorithm
	with any linearizable implementation of atomic registers
	(such as the one that we give for m\&m systems in Section~\ref{algo}),
	we obtain a randomized consensus algorithm that does work against a strong adversary.
So, from Theorem~\ref{thm-algo} in Section \ref{algo} and Theorem 20\RMP{this 20 is hardwired: check!} in~\cite{hadzilacos2020randomized},
	we have:

\begin{theorem}\label{thmconsup}
Let $\mathcal{S}_{L}$ be the general m\&m system induced by
	a bag $L = \{S_1,\dots ,S_m\}$ of subsets of $\Pi = \{ p_1 , p_2 , \ldots, p_n \}$.
By replacing the atomic SWMR registers used by the randomized consensus algorithm given in~\cite{AH1990consensus}
	with the linearizable implementation of such registers for system $\mathcal{S}_{L}$ given in Section~\ref{algo},
	we obtain an algorithm that solves randomized consensus in $\mathcal{S}_L$ 
	and tolerates up to $t_L$ process crashes.
\end{theorem}

It is worth noting that \cite{hadzilacos2020randomized} proved
	that the AH algorithm does not need register atomicity or linearizability to work:	
	in fact this algorithm works against a strong adversary
	even with \emph{regular} SWMR registers~\cite{Lam86}.
In contrast to atomic SWMR registers,
	each operation of a regular SWMR register spans an interval that starts with an \emph{invocation}
	and terminates with a \emph{response}.
Moreover, in contrast to linearizable implementations of atomic SWMR registers,
	a regular SWMR register satisfies only Property~\ref{p1} but not Property~\ref{p2} (Section~\ref{model}), and so
	it allows ``new-old'' inversions~\cite{Lam86}.

\subsection{Lower bound}\label{conslb}

A fault-tolerance lower bound on solving consensus in uniform m\&m systems
	was given in~\cite{aguilera-etal18} (Theorem 4.4).
A simple generalization of this result
	shows that the randomized consensus algorithm of Theorem~\ref{thmconsup} is optimal in the number
	of process crashes that it can tolerate in general m\&m systems.
More precisely:

\begin{theorem}\label{thmconsdown}
Let $\mathcal{S}_{L}$ be the general m\&m system induced by
	a bag $L = \{S_1,\dots ,S_m\}$ of subsets of $\Pi = \{ p_1 , p_2 , \ldots, p_n \}$.
If $t_L +1<n$ processes may crash, randomized consensus can not be solved in $\mathcal{S}_L$.
\end{theorem}

\begin{proof} [Proof Sketch]
As in~\cite{aguilera-etal18}, the proof for general general m\&m systems is by a standard partition argument. 
Suppose, for contradiction, that there is a randomized consensus algorithm $\cal A$
	that tolerates $t = t_L+1 <n $ process crashes in $\mathcal{S}_L$ (*).
Since $t>t_L$,
	by the Definition~\ref{def-tl} of $t_L$ there are two disjoint subsets $P$ and $P'$
	of $\Pi$,
	of size $n-t$ each, such that:
	no set $S_i$ in $L$ contains both a process in $P$ and a process in $P'$~(**).
Since $t<n$ each of $P$ and $P'$ contains at least one process.
Since $P$ and $P'$ are disjoint, the sets $P$, $P'$, and $Q=\Pi-(P\cup P')$
	form a partition of $\Pi$ (see Figure~\ref{gsx2}).
	
Consider the following execution of $\cal A$.
Processes in $Q$ take no steps (they crash at the start of this execution).
Processes in $P$ and $P'$ propose $0$ and $1$, respectively.
Processes in $P$ ``think'' that all the processes $P' \cup Q$ are crashed from the start,
while processes in $P'$ ``think'' that all the processes $P \cup Q$, because:
	\begin{itemize}
	\item Each of $P$ and $P'$ contains $n-t$ processes and up to $t$ process can crash in $\mathcal{S}_L$.

	\item All the messages between processes in $P$ and $P'$ are delayed, and

	\item by (**):
	\begin{itemize}
	\item 	 no value written by any process in $P'$ on a shared register can be read by any process in $P$.
	\item 	 no value written by any process in $P$ on a shared register can be read by any process in $P'$.
	\end{itemize}

	\end{itemize}

Since the consensus algorithm $\cal A$ tolerates $t$ crashes and terminates with probability~1,
	every process in $P$ and $P'$ eventually decides $0$ and $1$,
	respectively (all the delayed messages between them are received only after they decide);
	this violates the Agreement property of consensus.
\end{proof}

\section{Number of RDMA connections versus fault-tolerance degree}

Recall that in a pure message-passing systems with $n$ nodes,
	one can implement a SWMR atomic register,
	and solve randomized consensus, for up to $\lceil \frac{n}{2}\rceil -1$ crashes;
	so obtaining this degree of fault-tolerance does not require any shared memory or RDMA connection.
This raises the following question:
	what is the minimum number of RDMA connections required to tolerate \emph{more than} $\lceil \frac{n}{2}\rceil -1$ failures in a uniform m\&m system?\footnote{Recall that such a system is modeled by a graph $G$ where nodes represent processes and each edge between two processes represents an RDMA connection between these processes
	which allows them to share some SWMR registers.}
In this section we show that
	$m$ RDMA connections are necessary and sufficient to tolerate $m$ process crashes,
	for every $m$ such that $\lceil \frac{n}{2}\rceil -1 < m \leq n-1$.

\begin{lemma}\label{bl1}
For all $n \ge 2$, and every undirected graph $G$ with $n$ nodes and $\lceil \frac{n}{2} \rceil-1$ edges,
	we can partition the nodes of $G$ into two sets of nodes $P$ and $\overline{P}$ of size
	$\lfloor \frac{n}{2} \rfloor$ and $\lceil \frac{n}{2} \rceil$, respectively,
	such that there is no edge between any node of $P$ and any node of $\overline{P}$.
\end{lemma}
\begin{proof}
We prove the theorem for the two possible cases:
	$n$ is even (Claim~\ref{bc1}) and $n$ is odd (Claim~\ref{bc2}).

\begin{claim}\label{bc1}
For all $k \ge 1$, and every undirected graph $G$ with $n = 2k$ nodes and $\lceil \frac{n}{2} \rceil-1 = k-1$ edges,
	we can partition the nodes of $G$ into two sets of nodes $P$ and $\overline{P}$ of size
	$\lfloor \frac{n}{2} \rfloor = \lceil \frac{n}{2} \rceil = k$,
	such that there is no edge between any node of $P$ and any node of $\overline{P}$.
\end{claim}
\begin{proof}
We prove this by induction on $k$.

\textsc{Base case:} $k=1$. Clearly, for any graph with $n=2$ nodes and $0$ edges, the two nodes can be partitioned so that there is no edge between them.

\textsc{Inductive step:}
Let $k \ge 1$.
Assume that for every undirected graph with $2k$ nodes and $k-1$ edges,
	the nodes can be partitioned into sets $P$ and $\overline{P}$ of size $k$ each,
	with no edge between them (Induction Hypothesis).
We now show for every undirected graph with $2k+2$ nodes and $k$ edges,
	the nodes can be partitioned into sets $P$ and $\overline{P}$ of size $k+1$ each,
	with no edge between them

Let $G$ be any undirected graph with $2k+2$ nodes and $k$ edges.
Since each edge can ``decrease'' the number of singleton nodes (i.e., nodes with degree 0) by at most 2,
	it is clear that $G$ has at least two singleton nodes, say $u$ and $v$.

\textbf{Case 1}: $G$ contains no cycle.
Since $G$ has $k \ge 1$ edge and no cycle, it must have at least one node $a$ with degree 1.
Let $(a,b)$ be the (only) edge incident to $a$ in $G$.
Let $G'$ be the graph obtained by removing nodes $u$ and $v$ and the edge $(a,b)$ from $G$. 
So $G'$ has $2k$ nodes and $k-1$ edges.

By our Induction Hypothesis,
	the nodes $G'$ can be partitioned into sets $P'$ and $\overline{P'}$ of size $k$ each,
	 with no edge between them.
There are two subcases:
	
\emph{Case (i):} nodes $a,b$ are in the same partition.
Without loss of generality, suppose $a,b$ are in $P'$.
We partition the nodes of $G$ into sets $P$ and $\overline{P}$ as follows:
	$P = P' \cup \{u\}$ and $\overline{P} =\overline{P'} \cup \{v\}$.
Clearly, $P$ and $\overline{P}$ are of size $k+1$ each.
Furthermore, $G$ has no edge between $P$ and $\overline{P}$ since
	(1) by Induction Hypothesis, $G'$ has no edge between
	$P'$ and $\overline{P'}$,
	(2) both $a,b$ are in $P'$, and
	(3) $u$ and $v$ are singleton nodes.	
	
\emph{Case (ii):} nodes $a,b$ are not in the same partition.
Without loss of generality, assume node $a$ is in $P'$ and node $b$ is in $\overline{P'}$.
We partition the nodes of $G$ into sets $P$ and $\overline{P}$ as follows:
	$P = P'  \cup \{u,v\} - \{a\}$ and $\overline{P} =\overline{P'} \cup \{a\}$.
Clearly, $P$ and $\overline{P}$ are of size $k+1$ each.
Furthermore, $G$ has no edge between $P$ and $\overline{P}$ since
	(1) by Induction Hypothesis, $G'$ has no edge between
	$P'$ and $\overline{P'}$,
	(2) both $a,b$ are in $P'$, and
	(3) $u$ and $v$ are singleton nodes.	
	 
Therefore, in all subcases of Case 1,
	the nodes of $G$ can be partitioned into sets of nodes $P$ and $\overline{P}$ of size $k+1$ each, with no edge between them. 
	
\textbf{Case 2}: $G$ contains cycles. 
Let $(a,b)$ be any edge in a cycle.
Let $G'$ be the graph obtained by removing nodes $u$ and $v$ and the edge $(a,b)$ from $G$. 
So $G'$ has $2k$ nodes and $k-1$ edges.

By our Induction Hypothesis,
	the nodes $G'$ can be partitioned into sets $P'$ and $\overline{P'}$ of size $k$ each,
	 with no edge between them.
Since nodes $a,b$ are in a cycle in $G$,
	after removing $(a,b)$,
	there is still a path from $a$ to $b$ in $G'$.
Since $G'$ has no edge between any node of $P'$ and any node of $\overline{P'}$,
	nodes $a,b$ must be in the same partition of nodes, say $P'$.
	
The proof now proceeds as in Case 1(i).
We partition the nodes of $G$ into sets
	$P = P' \cup \{u\}$ and $\overline{P} =\overline{P'} \cup \{v\}$
	of size $k+1$ each.
Clearly $G$ has no edge between $P$ and $\overline{P}$ because:
	(1) by Induction Hypothesis, $G'$ has no edge between
	$P'$ and $\overline{P'}$,
	(2) both $a,b$ are in $P'$, and
	(3) $u$ and $v$~are~singleton~nodes.	

So in all cases,	the nodes of $G$ can be partitioned into sets of nodes $P$ and $\overline{P}$ of size $k+1$ each, with no edge between them. 
\end{proof}	
	
\begin{claim}\label{bc2}
For all $k \ge 1$, and every undirected graph $G$ with $n = 2k+1$ nodes and $\lceil \frac{n}{2} \rceil-1 = k$ edges,
	we can partition the nodes of $G$ into two sets of nodes $P$ and $\overline{P}$ of size
	$\lfloor \frac{n}{2} \rfloor = k $ and $\lceil \frac{n}{2} \rceil = k+1$, respectively,
	such that there is no edge between any node of $P$ and any node of $\overline{P}$.
\end{claim}

\begin{proof}
Let $k \ge 1$ and $G$ be any undirected graph with $2k+1$ nodes and $k$ edges.
	We now show that
	the nodes of $G$ can be partitioned into two sets $P$ and $\overline{P}$ of size
	$k $ and $k+1$, respectively,
	such that there is no edge between them.

Let $G'= G +\{x\}$ where $x$ is a new singleton node.
Clearly, 
	 $G'$ has $2k+2$ nodes and $k$ edges.
By Claim~\ref{bc1},
	 the nodes of $G'$ can be partitioned into two sets $P'$ and $\overline{P'}$ of size $k+1$ each, with no edge between them in $G'$.
Without loss of generality, assume node $x$ is in $P'$.
We partition the nodes of $G$ into the two sets
	$P = P' - \{x\}$ and $\overline{P} =\overline{P'}$,
	of size $k$ and
	$k+1$, respectively.
Since $G'$ has no edge between
	$P'$ and $\overline{P'}$,
	$G$ has no edge between $P$ and $\overline{P}$.
\end{proof}
The lemma follows from Claim~\ref{bc1}) ($n$ is even) and Claim~\ref{bc2} ($n$ is odd).
\end{proof}

\begin{theorem}\label{bt1}
For all $n \ge 2$,  every undirected graph $G$ with $n$ nodes and $\lceil \frac{n}{2} \rceil-1$ edges has $t_G \leq \lceil \frac{n}{2}\rceil -1$.
\end{theorem}

\begin{proof}
Consider any undirected graph $G$ with $n$ nodes and $\lceil \frac{n}{2} \rceil-1$ edges where $n \ge 2$.
By Lemma~\ref{bl1},
	the nodes of $G$ can be partitioned into two sets $P$ and $\overline{P}$ of size at least $\lfloor \frac{n}{2} \rfloor$ each
	such that $G$ has no edge between them.
So $G^2$ does not have any edge between any node in $P$ and any node in $\overline{P}$.
By the definition of $t_G$ (Definition~\ref{doesnothaveone}),
	this implies that $t_G < n-\lfloor \frac{n}{2} \rfloor$.
Therefore  $t_G \le n-\lfloor \frac{n}{2} \rfloor -1 = \lceil \frac{n}{2}\rceil -1$.
\end{proof}

\begin{theorem}\label{rdma-min}
For all $n\geq 2$ and
	all $\lceil \frac{n}{2}\rceil -1 \leq m \leq n-1$:

\begin{enumerate}
\item Every graph $G$ with $n$ nodes and $m$ edges has $t_G \leq m$.

\item Some graph $G$ with $n$ nodes and $m$ edges has $t_G = m$.

\end{enumerate}

\end{theorem}

\begin{proof}

Let $n\ge 2$.

\begin{enumerate}

\item We prove Part 1 by induction on the number of edges $m$.

\textsc{Base case:} $m = \lceil \frac{n}{2} \rceil -1$. 
By Theorem~\ref{bt1}, 
	 every undirected graph $G$ with $n$ nodes and $\lceil \frac{n}{2} \rceil-1$ edges has $t_G \leq \lceil \frac{n}{2}\rceil -1$.
	
\textsc{Inductive step:}
Let $\lceil \frac{n}{2} \rceil -1 \leq k < n-1$.
Assume every graph $G$ with $n$ nodes and $k$ edges has $t_G \leq k$ (Induction Hypothesis).
Consider any graph $G =(V,E)$ with $n$ nodes and $k+1$ edges.
We must show that $G$ has $t_G \leq k+1$.
To prove this, by the definition of $t_G$ (Definition~\ref{doesnothaveone}), we must show that: (*) $G$ has two disjoint sets of nodes $P$ and $Q$ of size $n-(k+2)$ each such that $G^2$ has no edge between a node in $P$ and a node in $Q$.

For any set of nodes $S$ of $G$, let $(\delta S)_G$ be the \emph{set of neighbours of $S$ in $G$},
i.e., $(\delta S)_G = \{ v\in V: (u,v)\in E, u\in S,v \notin S\}$.
Note that proving (*) is equivalent to proving that $G$ has
	 two disjoint sets of nodes $P$ and $Q$ of size $n-(k+2)$ each such that
	$P \cup (\delta P)_G$ and $Q \cup (\delta Q)_G$ are disjoint.

\begin{observation}\label{bo1}
For any two disjoint sets of nodes $P'$ and $Q'$ in $G$, 
	if $P' \cup (\delta P')_G$ and $Q' \cup (\delta Q')_G$ are disjoint,
	then for any subsets $P \subseteq P'$ and $Q \subseteq Q'$,
	 $P \cup (\delta P)_G$ and $Q \cup (\delta Q)_G$ are also disjoint.
\end{observation}

Let $e$ be any edge of graph $G$ (this edge exists because, $n\ge 2$, $k \ge \lceil \frac{n}{2} \rceil -1 \ge 0$, and $G$ has $k+1$ edges).
Let $G'$ be the graph obtained by removing edge $e$ from $G$.
Thus, $G'$ has $n$ nodes and $k$ edges.
By the induction hypothesis, $t_{G'} \leq k$.
So, by the definition of $t_G$, $G'$ has two disjoint sets of nodes $P'$ and $Q'$ of size $n-(k+1)$ each such that $G'^2$ has no edge 	between a node in $P'$ and a node in $Q'$.
This implies $P' \cup (\delta P')_{G'}$ and $Q'\cup (\delta Q')_{G'}$ are disjoint (*).
There are two cases:

\textbf{Case 1}: $e$ is between two nodes in $P'$ or between two nodes in $Q'$.
In this case it is clear that $(\delta P')_G =(\delta P')_{G'} $ and $(\delta Q')_G = (\delta Q')_{G'}$.
Since by (*), $P' \cup (\delta P')_{G'}$ and $Q'\cup (\delta Q')_{G'}$ are disjoint,
	$P' \cup (\delta P')_G$ and $Q'\cup (\delta Q')_G$ are disjoint.
Let $P$ and $Q$ be any two subsets of $P'$ and $Q'$, respectively, of size $n-(k+2)$ each.
By Observation~\ref{bo1},
	$P \cup (\delta P)_G$ and $Q \cup (\delta Q)_G$ are also disjoint.
	
\textbf{Case 2}: $e$ is not between two nodes in $P'$ or between two nodes in $Q'$.
So $e$ connects at most one node in $P'$ and at most one node in $Q'$.
Thus, since $|P'|=|Q'|=n-(k+1)$,
	there exist subsets of nodes $P \subseteq P'$, $Q \subseteq Q'$
	such that $|P|=|Q|=n-(k+2)$ and no endpoint of $e$ is in $P$ or $Q$.
By (*) and Observation~\ref{bo1},
	$P \cup (\delta P)_{G'}$ and $Q \cup (\delta Q)_{G'}$ are disjoint.
Since $G$ differs from $G'$ only by having the extra edge $e$,
	and no endpoint of $e$ is in $P$ or $Q$,
	it is clear that $(\delta P)_G = (\delta P)_{G'}$ and $(\delta Q)_G =  (\delta Q)_{G'}$.
So, $P \cup (\delta P)_G$ and $Q \cup (\delta Q)_G$ are disjoint.

Since in all possible cases
	$G$ has
	 two disjoint sets of nodes $P$ and $Q$ of size $n-(k+2)$ each such that
	$P \cup (\delta P)_G$ and $Q \cup (\delta Q)_G$ are disjoint, 
 	Part 1 holds.
	
\item Let $m$ be such that $\lceil \frac{n}{2}\rceil -1 \leq m \leq n-1$.
To show Part 2,
	we now describe a graph $G$ with $n$ nodes and $m$ edges that has $t_G=m$.

The $n$ nodes of $G$ are $v_0, v_1 ,..., v_{n-1}$,
	and $G$ has an edge between $v_0$ and $v_i$ for every $1 \le i \le m$
	(so there are $n - m -1$ singleton nodes, namely  $v_{m+1},..., v_{n-1}$).

\begin{claim}\label{bc3}
$t_G \geq m$.
\end{claim}

\begin{proof}
By the definition of $t_G$, we must prove that for any two disjoint sets
	of nodes $P$ and $Q$ of size $n-m$ each, $G^2$
	has an edge between a node in $P$ and a node in $Q$.
	
Consider any two disjoint sets
	of nodes $P$ and $Q$ of size $n-m$ each.
Since there are only $n- m -1$ nodes in $\{ v_{m+1},..., v_{n-1} \}$,
	each of $P$ and $Q$ must have at least one node in $\{ v_0, ..., v_m \}$; say $v_i \in P$ and $v_j \in Q$.
Since $G$ has an edge between $v_0$ and $v_k$ for every $1 \le k \le m$,
	$G^2$ has an edge between $v_i \in P$ and $v_j \in Q$, as we needed to show.
\end{proof}
By Claim~\ref{bc3} and Part 1,
	$t_G = m$.
\end{enumerate}
\end{proof}

We can now answer the following question:
	what is the minimum number of RDMA connections required to tolerate \emph{more than} $\lceil \frac{n}{2}\rceil -1$ failures in a uniform m\&m system?
The answer is given by combining Theorem~\ref{rdma-min} with Theorems~\ref{thm44} and~\ref{thmcons0}:
	$m$ RDMA connections are necessary and sufficient to tolerate $m$ process crashes,
	for every $m$ such that $\lceil \frac{n}{2}\rceil -1 < m \le n-1$.
More precisely:

\begin{theorem}
Let $n\geq 2$ and
	 $\lceil \frac{n}{2}\rceil -1 < m$.

\begin{enumerate}
\item If $m \le n-1$, for some graph $G$ with $n$ nodes and $m$ edges,
	in the uniform m\&m system $\mathcal{S}_{G}$ induced by $G$:

	- there is an $m$-tolerant implementation of an atomic SWMR register for any writer $w$.
	
	- there is an $m$-tolerant randomized consensus algorithm.

\item  If $m < n-1$, for every graph $G$ with $n$ nodes and $m$ edges,
	in the uniform m\&m system $\mathcal{S}_{G}$ induced by $G$:

	- there is no $(m+1)$-tolerant implementation of an atomic SWMR register for any writer $w$.
	
	- there is no $(m+1)$-tolerant randomized consensus algorithm.

\end{enumerate}

\end{theorem}

\section{Concluding remarks}
Hybrid systems that combine message passing and shared memory
	have long been a subject of study in the systems community
	\cite{treadmarks,barrelfish,munin,DGY04,argodsm,alewife,grappa,shasta}.
To the best of our knowledge, however, such systems have only recently been examined from a theoretical point of view.
Aguilera \emph{et al.} gave a rigorous model for hybrid systems, namely the m\&m model,
	and studied how the combination of message passing
	and shared memory can be harnessed to improve solutions to certain fundamental problems:
In particular, they show that, compared to a pure message-passing system, a hybrid system can improve the fault tolerance of randomized consensus algorithms and reduce the synchrony necessary to elect a leader \cite{aguilera-etal18}.
A more recent paper by Aguilera \emph{et al.} extends the m\&m model to Byzantine failures,
	and shows how to improve the inherent trade-off between fault tolerance
	and performance for consensus, for both Byzantine and crash failures~\cite{AguileraBGMZ19}.
The present paper is another contribution to the theoretical study of hybrid systems:
	whereas the well-known ABD algorithm
	implements an atomic SWMR register
	with optimal fault tolerance in a pure message-passing system \cite{attiya1995sharing},
	here we implement such registers with optimal fault tolerance in m\&m systems.
We also show how to solve randomized consensus with optimal fault tolerance in such systems.
Extending our results to hybrid systems with Byzantine failures
	is a subject for future research.
	
Another possible extension to this work regards the design of uniform m\&m systems that maximize the fault-tolerance
	under some constraints on RDMA connections.
In this paper, we proved that to implement SWMR registers (or solve randomized consensus) in uniform m\&m systems,
	$m$ RDMA connections are necessary
	and sufficient for tolerating $m$ process crashes.
	In the ``sufficient'' part of our proof, however, there is a process that has an RDMA connection
	to every other process in the system;
	the corresponding graph $G$ is a ``star'' graph where one node has degree $n-1$ and every other node has degree 1.
In practice it is often desirable to limit the number of RDMA connections per process to some $k$.
So this raises the following question:
	what is the maximum fault tolerance that can be achieved in
	uniform m\&m systems $S_G$ induced by graphs $G$ \emph{of degree $k$}?
For $k=1$, it is easy to see that the m\&m system $S_G$ induced by the graph $G$ that consists of pairs of connected nodes is optimal,
	but its fault tolerance is quite low:
	with $n$ nodes, $S_G$ tolerates up to $t_G =  {n}/{2}$ process crashes if $n =2(2i+1)$ for some $i\ge0$,
	but only $t_G =  \lceil {n}/{2}\rceil -1$ crashes otherwise (which is no better than a pure message-passing system).
For $k=2$, the system $S_G$ induced by the graph $G$
	consisting of a simple cycle of $n$ nodes is optimal; with $n$ nodes,
	$S_G$ tolerates up to $t_G=\lceil n/2 \rceil +1$ process crashes (see Appendix~\ref{Two}).
The following is an open problem:
for each $k \ge 3$, find a graph $G$ of degree $k$
	that maximizes the number of process crashes
	tolerated by the induced m\&m system $S_G$.

\bibliographystyle{abbrv}
\bibliography{New}

\begin{appendices}

\section{Algorithm~\ref{algo1} is not strongly linearizable}\label{nsl}

We now prove that our implementation of atomic SWMR registers for m\&m systems given in Section \ref{algo}
	is \emph{not} strongly linearizable.
To do so, we show that the ABD algorithm~\cite{attiya1995sharing} that implements atomic SWMR registers for pure message-passing systems is not strongly linearizable; recall that the ABD algorithm
	is a special case of Algorithm~\ref{algo1}.

First recall the definition of \emph{strong linearizability} \cite{sl1}:
\begin{definition}
Let $\mathcal{H}$ be a prefix-closed set of histories. 
$\mathcal{H}$ is \emph{strongly linearizable} if there exists a function f mapping histories in $\mathcal{H}$ to sequential histories, such that
\begin{itemize}
\item for any $H \in \mathcal{H}$, $f(H)$ is a linearization of $H$, and
\item for any $G,H\in \mathcal{H}$, if $G$ is a prefix of $H$, then $f(G)$ is a prefix of $f(H)$.
\end{itemize}
\end{definition}

\begin{definition}
An implementation of a shared object type is strongly linearizable if
	the set of histories of the implementation is strongly linearizable.\footnote{In a history of an object implementation, we omit all steps other than the invocation and response steps on that object.}
\end{definition}

\begin{algorithm}[!ht]
\caption{ The ABD implementation of an atomic SWMR register writable by process~$w$ and readable by all processes in a message-passing system $ \mathcal{S}$, provided that at most $\lceil n/2\rceil -1$ processes~crash.}
\label{algoa}
~\\
\vspace{1mm}
$R[p]:$ local register writable and readable only by $p$ ;\\
\hspace*{1.1cm}initialized to $\langle 0 , \rinit \rangle$.

\vspace{5mm}

\begin{algorithmic}[1]
\Statex

\textsc{\wu($\langle sn_w,u \rangle$):} \Comment{executed by the writer $w$}
\Indent
\State $\textbf{send } \langle \textrm{W},\langle sn_w, u\rangle \rangle \textbf{ to every process } p\mbox{ in } \mathcal{S}$\label{communicatea}
\State \textbf{wait for} $\langle \textrm{ACK-W},sn_w \rangle \textbf{ messages from } \lceil \frac{n+1}{2}\rceil \textbf{ distinct processes }$\label{wait1a}
\State \Return
\EndIndent
\Statex
\Statex

\Comment{executed by every process $p \mbox{ in } \mathcal{S}$}\\
\textbf{upon receipt of a }$\langle \textrm{W},\langle sn_w, u\rangle \rangle$ \textbf{message from process $w$}:
\Indent
\State\label{wr1a} $\langle sn,-\rangle \gets R[p]$
\State\label{wr2a} \textbf{if} {$sn_w> sn$}\textbf{ then}
\Indent
\State\label{wr3a} $R[p]\gets \langle sn_w,u\rangle$
\EndIndent
\State\label{writeack} \textbf{send} $\langle\textrm{ACK-W}, sn_w \rangle$ \textbf{to process} $w$
\EndIndent
\Statex
\Statex

\textsc{\ru():}\Comment{executed by any process $q$}
\Indent
\State $sn_r\gets sn_r+1$
\State $\textbf{send } \langle \textrm{R}, sn_r\rangle \textbf{ to every process } p \mbox{ in } \mathcal{S}$\label{reada}
\State \textbf{wait for} $\langle\textrm{ACK-R}, sn_r, \langle -, - \rangle \rangle \textbf{ messages from } \lceil \frac{n+1}{2}\rceil  \textbf{ distinct processes }\label{wait2a}$
\State $\langle seq, val \rangle \gets \max\{ \langle r\_sn,r\_u \rangle~|$ received a $\langle\textrm{ACK-R}, sn_r,\langle r\_sn,r\_u \rangle \rangle$ message\} \label{adopta}
\State\label{callWUbyReadera} \textsc{\wu}{$(\langle seq, val \rangle)$}
\State \label{readreturn}\Return $\langle seq, val \rangle$
\EndIndent
\Statex
\Statex

\Comment{executed by every process $p \mbox{ in } \mathcal{S}$}\\
\textbf{upon receipt of a} $ \langle \textrm{R}, sn_r\rangle$ \textbf{message from a process $q$}:
\Indent
\State$\langle r\_sn,r\_u\rangle \gets R[p]$\label{readlocala}
\State \textbf{send} $\langle \textrm{ACK-R}, sn_r,\langle r\_sn,r\_u \rangle\rangle$ \textbf{to process} $q$\label{sendtoreadera}
\EndIndent
\Statex
\end{algorithmic}
\end{algorithm}

\begin{theorem}
The ABD implementation of an atomic SWMR register in pure message-passing systems (shown in Algorithm~\ref{algoa}) is \emph{not} strongly linearizable. 
\end{theorem}
\begin{proof}
Consider a pure message-passing system $\mathcal{S}$ with $3$ processes,
	namely, $w,p,q$.
Let $\REG$ be the atomic SWMR register implemented by Algorithm \ref{algoa} in $\mathcal{S}$.
$\REG$ can be written by $w$ and read by all processes of $\mathcal{S}$.

Let $\mathcal{H}$ be the set of histories of the Algorithm \ref{algoa} 
	(in these histories we omit all steps other than the invocations and responses of read and write operations on $\REG$).
To prove that Algorithm~\ref{algoa} is not strongly linearizable, 
		we show that $\mathcal{H}$ is not strongly linearizable. 
More precisely,
	we prove that for any function $f$ that maps histories in $\mathcal{H}$ to sequential histories,
	there exist histories $G,H\in \mathcal{H}$ such that $G$ is a prefix of $H$ but $f(G)$ is not a prefix of $f(H)$.
	
\begin{figure}[!ht]
\centering
\includegraphics[width=0.6\textwidth]{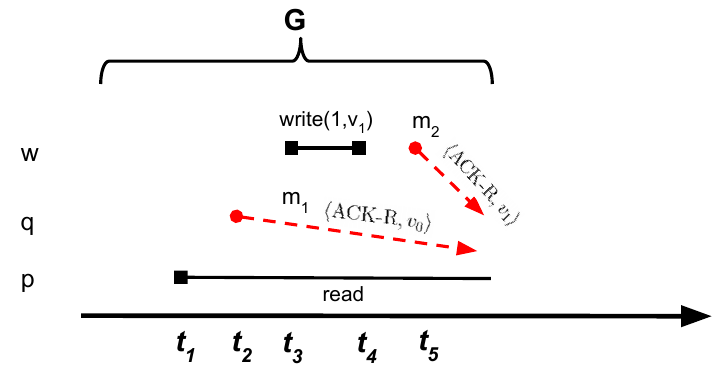}
  \caption{History $G$}
  \label{Hg}
\end{figure}

Let $f$ be a function that maps histories in $\mathcal{H}$ to sequential histories.
Consider the following history $G \in \mathcal{H}$ (shown in Figure~\ref{Hg}):
\begin{itemize}
\item Initially, $\REG$ contains $v_0$, and so all the local registers $R[-]$ contain the value $v_0$.
\item At time $t_1$,
	process $p$ starts an operation $\R$ to read $\REG$.
According to line \ref{reada} of Algorithm \ref{algoa},
	$p$ first sends $\langle \textrm{R}, sn_r\rangle$ to all processes,
	then:
	\begin{itemize}
	\item $p$ receives $\langle \textrm{R}, sn_r\rangle$ from itself,
	  reads $\langle 0, v_0 \rangle$ from $R[p]$ (line~\ref{readlocala}),
		and sends $\langle  \textrm{ACK-R}, sn_r,\langle 0, v_0 \rangle\rangle$ to itself (line~\ref{sendtoreadera}).
And so $p$ receives $\langle \textrm{ACK-R}, sn_r,\langle 0, v_0 \rangle\rangle$ from itself.
	\item let $m_0$ denote the message $\langle \textrm{R}, sn_r\rangle$ from $p$ to $w$;  delay $m_0$.
Since $w$ does not receive $m_0$, 
	$w$ takes no step.
	\item $q$ receives the message $\langle \textrm{R}, sn_r\rangle$ from $p$ and
	 reads $\langle 0, v_0 \rangle$ from $R[q]$ (line \ref{readlocala}).
Then $q$ sends back $\langle \textrm{ACK-R}, sn_r,\langle 0, v_0 \rangle\rangle$ to $p$ (line \ref{sendtoreadera}), 
	say at time $t_2$.
Let $m_1$ denote the message $\langle \textrm{ACK-R}, sn_r,\langle 0, v_0\rangle\rangle$ from $q$ to $p$ and delay $m_1$.
	\end{itemize}

\item At some time $t_3 > t_2$,
	the writer $w$ starts an operation $\W$
	to write the value $v_1$ into $\REG$ with sequence number $1$, 
	for some $v_1 \neq v_0$.
Process $w$ first sends the message $\langle \textrm{W},\langle 1, v_1\rangle \rangle$ to all processes (line \ref{communicatea}) including itself, but the message to $p$ is delayed.
Processes $w$ and $q$ receive $\langle \textrm{W},\langle 1, v_1\rangle \rangle$ from $w$,
	and since $R[w]$ and $R[q]$ contain $\langle 0, v_0\rangle$,
	by line~\ref{wr2a} of Algorithm~\ref{algoa},
	both $w$ and $q$ write $\langle 1, v_1\rangle$ to $R[w]$ and $R[q]$ respectively (line~\ref{wr3a}), and they
	send $\langle \textrm{ACK-W}, 1\rangle$ to $w$ (line~\ref{writeack}).
So $w$ receives $\langle \textrm{ACK-W}, 1\rangle$ from $w$ and $q$.
By line \ref{wait1a},
	the write operation $\W$ terminates,
	say at time $t_4$.
\item After time $t_4$,
	$w$ receives the delayed message $m_0$ from $p$.
Since now $R[w]$ contains $\langle 1, v_1\rangle$,
	$w$ reads $\langle 1, v_1\rangle$ in line~\ref{readlocala}.
And so $w$ sends $\langle \textrm{ACK-R}, sn_r,\langle  1, v_1  \rangle\rangle$ to $p$ (line \ref{sendtoreadera}), say at time $t_5$.
Let $m_2$ denote the message $\langle \textrm{ACK-R}, sn_r,\langle  1, v_1  \rangle\rangle$ from $w$ to $p$; delay $m_2$.
\end{itemize}

Note that in $G$,
	messages $m_1 = \langle \textrm{ACK-R}, sn_r,\langle 0, v_0 \rangle\rangle$ and $m_2 = \langle \textrm{ACK-R}, sn_r,\langle  1, v_1  \rangle\rangle$ are sent to $p$ but not yet received by $p$.
	As we will see,
		the order $p$ will receive these two messages determines the value that $p$ will read,
		 and hence determines how $p$'s read is linearised with respect to $w$'s write.

Since the write operation $\W$ terminates in $G$ and $f(G)$ is a linearisation of $G$,
	$\W$ is in $f(G)$.
Since the read operation $\R$ is concurrent with $\W$,
there are two cases: (1) $\R$ is before $\W$ in $f(G)$,
(2) $\R$ is not before $\W$ in $f(G)$.

\begin{figure}[!ht]
\centering
\includegraphics[width=0.6\textwidth]{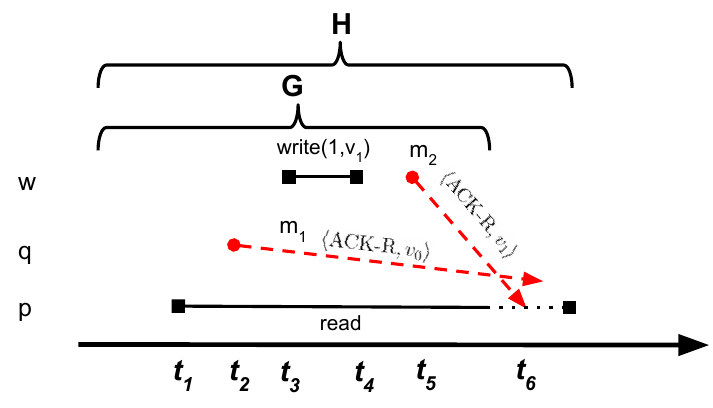}
  \caption{History $H$ of Case 1}
  \label{Hh1}
\end{figure}

\textbf{Case 1}: $\R$ is before $\W$ in $f(G)$.
Consider the following history $H \in \mathcal{H}$ (Figure~\ref{Hh1}):
\begin{itemize}
\item $H$ is an extension of $G$, i.e., $G$ is a prefix of $H$.
\item  At time $t_6> t_5$,
	$p$ receives the delayed message $m_2$ from $w$. 
Since $p$ receives $\langle  0, v_0 \rangle$ from itself and receives $\langle  1, v_1  \rangle$ from $w$,
 	by line \ref{adopta},
	$p$ selects $\langle  1, v_1  \rangle$, 
	writes back $\langle  1, v_1  \rangle$ in line~\ref{callWUbyReadera} 
	and returns $\langle 1, v_1  \rangle$ in line~\ref{readreturn},
i.e.,
	the read operation $\R$ of $p$ returns $v_1$.
\end{itemize}

Since the read operation $\R$ of $p$ returns $v_1$ in $H$, and $f(H)$ is a linearisation of $H$, 
		by Property~\ref{p1} of linearizable atomic SWMR register implementation,
	 	$\R$ is after $\W$ in $f(H)$.
However, since, by assumption,
	 $\R$ is before $\W$ in $f(G)$,
	$f(G)$ is not a prefix of $f(H)$.

\begin{figure}[!ht]
\centering
\includegraphics[width=0.6\textwidth]{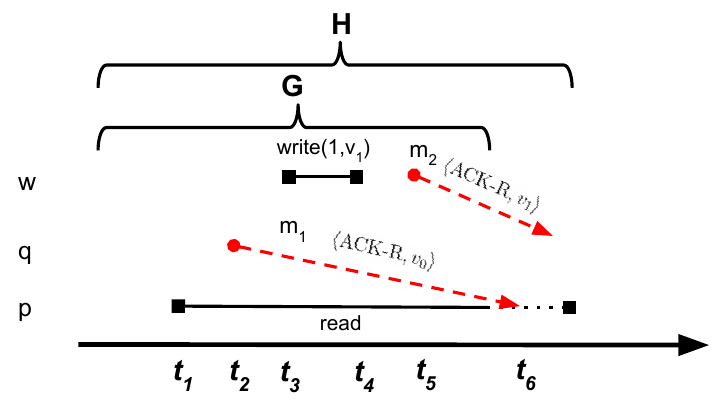}
  \caption{History $H$ of Case 2}
  \label{Hh2}
\end{figure}
\textbf{Case 2}: $\R$ is not before $\W$ in $f(G)$.
Consider the following history $H \in \mathcal{H}$ (Figure~\ref{Hh2}):
\begin{itemize}
\item $G$ is a prefix of $H$.
\item  At time $t_6 > t_5$,
	$p$ receives the delayed message $m_1$ from $q$. 
Since $p$ receives $\langle  0, v_0\rangle$ from both itself and $q$,
 	by line \ref{adopta},
	$p$ selects $\langle  0, v_0  \rangle$,
	 writes back $\langle  0, v_0  \rangle$ in line~\ref{callWUbyReadera},
	and returns $\langle 0, v_0  \rangle$ in line~\ref{readreturn},
 	 i.e.,
	the read operation $\R$ of $p$ returns $v_0$.
\end{itemize}

Since the read operation $\R$ of $p$ returns $v_0$ in $H$, and $f(H)$ is a linearisation of $H$,
	by Property~\ref{p1} of linearizable atomic SWMR register implementation,
	$\R$ is before $\W$ in $f(H)$.
However, since, by assumption, 
	$\R$ is not before $\W$ in $f(G)$,
	$f(G)$ is not a prefix of $f(H)$.

So in both cases, there is a history $H \in \mathcal{H}$
	such that $G$ is a prefix of $H$ but $f(G)$ is not a prefix of $f(H)$.
Therefore
	the theorem holds.
\end{proof}
\section{Optimal uniform m\&m systems limited to 2 RDMA connections per process}\label{Two}

We show that in uniform m\&m systems with $n$ processes,
	if we limit the number of RDMA connections to only two per process,
	then the maximum number of process crashes that can be tolerated
	(for implementing a SWMR register or solving randomized consensus) is $\lceil n/2 \rceil +1$,
	and this can be achieved by connecting the $n$ processes into a simple cycle.
This follows from Theorems~\ref{thm44},~\ref{thmcons0}, and the theorem~below:

\begin{theorem}

\ 
\begin{enumerate}[(1)]

\item Every graph $G$ with $n$ nodes and degree 2 has $t_G \le \lceil n/2 \rceil +1$.

\item The graph $G$ that
	consists of a simple cycle of $n$ nodes has $t_G=\lceil n/2 \rceil +1$.
\end{enumerate}
\end{theorem}

\begin{proof}
~
\newline
(1) Let $G$ be a graph with $n$ nodes and degree $2$.
To prove $t_G < \lceil n/2 \rceil +2$,
	we show that there are two disjoint subsets of nodes of size $n - (\lceil n/2 \rceil +2) = \lfloor n/2 \rfloor -2$ each 
	such that $G^2$ has no edge between~them.
	
First we partition the nodes in $G$ into two subsets $P$ and $Q$ 
	such that $|P| = \lfloor n/2 \rfloor$, 
	$|Q| = \lceil n/2 \rceil$, and there are \emph{at most two} edges in $G$ between nodes in $P$ and nodes in $Q$.
This can be done as follows.
Let $C_1, C_2, \ldots , C_\ell$ be the connected components of $G$, and let $n_i$ be the number of nodes of $C_i$ for $1 \le i \le \ell$.
There must be a component $C_j$ such that either:

\begin{enumerate}[(i)]
\item $n_1 + n_2 +  \cdots + n_{j-1} = n_j + n_{j+1} +  \cdots + n_{\ell}$, or

\item $n_1 + n_2 +  \cdots + n_{j-1} < n_j + n_{j+1} +  \cdots + n_{\ell}$ and  $n_1 + n_2 +  \cdots + n_{j-1} + n_j  > n_{j+1} +  \cdots + n_{\ell}$

\end{enumerate}

In the case (i), $P$ is the set of nodes in $C_1 , C_2 ,  \ldots , C_{j-1}$,
	and $Q$ is the set of nodes in $C_j , C_{j+1} ,  \ldots , C_{\ell}$.
Clearly $|P| = |Q| = n/2$, and there are no edges between the nodes of $P$ and the nodes of $Q$.

In the case (ii), it is easy to see that it is possible to split the $n_j$ nodes
	of component $C_j$ into $n_j^1$ and $n_j^2$ nodes such that
	$n_1 + n_2 +  \cdots + n_{j-1} + n_j^1 =  \lfloor n/2 \rfloor$ and
	$n_j^2 + n_{j+1} +  \cdots + n_{\ell} = \lceil n/2 \rceil$.
Furthermore, since $G$ has degree 2,
	$C_j$ is either a chain or a cycle, and so we can select
	the $n_j^1$ and $n_j^2$ nodes from $C_j$ such that
	$C_j$ has \emph{at most two edges} between these two sets of nodes.
Let $P$ be the set of nodes in $C_1 , C_2 ,  \ldots , C_{j-1}$ and the $n_j^1$ nodes from $C_j$,
	and $Q$ be the set of nodes in $C_j , C_{j+1} ,  \ldots , C_{\ell}$ and the $n_j^2$ nodes from $C_j$.
Clearly $|P| = \lfloor n/2 \rfloor$, $|Q| = \lceil n/2 \rceil$, and there are at most two edges
	between the nodes of $P$ and the nodes of $Q$.

Now we remove the endpoints of the edges between $P$ and $Q$, if such edges exist;
	note that this takes out at most two nodes from $P$ and two nodes from $Q$.
This gives two subsets $P'$ and $Q'$ of $P$ and $Q$
	such that:
	(a)~there are no edges between the nodes of $P'$ and the nodes of $Q'$, and
	(b)~$|P'| \ge \lfloor n/2 \rfloor-2$, $|Q'| \ge \lceil n/2 \rceil -2$.
Note that any node in $P'$ is at least 3 edges away from any node in $Q'$.
So no edge in $G^2$ connects a node in $P'$ and a node in $Q'$. 
Thus, there are two disjoint sets of nodes (namely, $P'$ and $Q'$)
	of size $\lfloor n/2 \rfloor-2$ each
	such that $G^2$ has no edge between them.
Therefore, $t_G < n - (\lfloor n/2 \rfloor-2) = \lceil n/2 \rceil +2$.

\vspace{2mm}
\noindent
(2) Consider the graph $G$ that consists of a simple cycle of $n$ nodes.
We show that $t_G \ge \lceil n/2 \rceil +1$.
For any subset $P$ of nodes of size $n-(\lceil n/2 \rceil +1) = \lfloor n/2 \rfloor -1$ in the cycle $G$,
	$P$ has at least two neighbours in $G$, i.e., $\delta P \ge 2$,
	and so $|P\cup \delta P| \ge \lfloor n/2 \rfloor +1$.
Thus, for every two sets $P$, $Q$ of nodes of size $\lfloor n/2 \rfloor -1$,
	$P\cup \delta P$ and $Q\cup \delta Q$ intersect.
This implies that $G^2$ has an edge between any two disjoint subsets of nodes of size $\lfloor n/2 \rfloor -1$.
Therefore,
	$t_G \ge n - (\lfloor n/2 \rfloor-1) = \lceil n/2 \rceil +1$.
\end{proof}

\end{appendices}
\end{document}